\documentclass[a4paper,twocolumn,accepted=2022-10-23]{quantumarticle}
\pdfoutput=1

\usepackage[utf8]{inputenc}
\usepackage[english]{babel}
\usepackage[T1]{fontenc}
\usepackage{amssymb,amsfonts,amsmath,amsthm}
\usepackage{graphicx}
\usepackage[numbers,sort&compress]{natbib} 	

\usepackage{color}
\usepackage{soul}

\usepackage[colorlinks=true]{hyperref}

\usepackage{ulem}

\newcommand{\bra}[1]{\langle {#1} \vert}

\newcommand{\ket}[1]{\vert {#1} \rangle}
\newcommand{\kket}[1]{\vert {#1} \rangle\!\rangle}

\newcommand{\proj}[1]{\vert {#1} \rangle\!\langle {#1} \vert}
\newcommand{\ketbra}[2]{\vert {#1} \rangle\!\langle {#2} \vert}

\newcommand{\Tr}[0]{{\mathrm{Tr}}}
\newcommand{\fcal}[0]{{\mathcal{F}}}

\newcommand{\ical}[0]{{\mathcal{I}}}
\newcommand{\ocal}[0]{{\mathcal{O}}}
\newcommand{\pcal}[0]{{\mathcal{P}}}
\newcommand{\lcal}[0]{{\mathcal{L}}}

\newcommand{\Jamiolkowski}[0]{{Jamio\l kowski}}

\makeatletter
\newcommand{\doublewidetilde}[1]{{%
  \mathpalette\double@widetilde{#1}%
}}
\newcommand{\double@widetilde}[2]{%
  \sbox\z@{$\m@th#1\widetilde{#2}$}%
  \ht\z@=.9\ht\z@
  \widetilde{\box\z@}%
}
\makeatother
\newcommand{\ttsmap}[1]{\doublewidetilde{\mbox{$\mathcal{#1}$}}}
\newcommand{\smap}[1]{\doublewidetilde{\mbox{$#1$}}}
\newcommand{\tmap}[1]{\widetilde{\mbox{$\mathcal{#1}$}}}
\newcommand{\tmapf}[1]{\widetilde{\mbox{$#1$}}}

\newtheorem{lem}{Lemma}
\newtheorem{thm}{Theorem}

\newtheorem{coll}{Corollary}

\theoremstyle{definition}
\newtheorem{rem}{Remark}

\begin{document}

\title{The quantum switch is uniquely defined by its action on unitary operations}

\author{Qingxiuxiong Dong}
 \orcid{0000-0002-2257-115X}
 \affiliation{Department of Physics, Graduate School of Science, The University of Tokyo, 7-3-1 Hongo, Bunkyo-ku, Tokyo 113-0033, Japan}
\author{Marco T{\'u}lio Quintino}
 \email{Marco.Quintino@lip6.fr}
 \orcid{0000-0003-1332-3477}
\affiliation{Sorbonne Universit\' {e}, CNRS, LIP6, F-75005 Paris, France}
 \affiliation{Institute for Quantum Optics and Quantum Information (IQOQI), Austrian Academy of Sciences, Boltzmanngasse 3, A-1090 Vienna, Austria}
 \affiliation{Vienna Center for Quantum Science and Technology (VCQ), Faculty of Physics, University of Vienna, Boltzmanngasse 5, A-1090 Vienna, Austria}
 \affiliation{Department of Physics, Graduate School of Science, The University of Tokyo, 7-3-1 Hongo, Bunkyo-ku, Tokyo 113-0033, Japan}
\author{Akihito Soeda}
 \email{soeda@nii.ac.jp}
 \orcid{0000-0002-7502-5582}
 \affiliation{Principles of Informatics Research Division, National Institute of Informatics, 2-1-2 Hitotsubashi, Chiyoda-ku,Tokyo 101-8430, Japan}
 \affiliation{Department of Informatics, School of Multidisciplinary Sciences, SOKENDAI (The Graduate University for Advanced Studies), 2-1-2 Hitotsubashi, Chiyoda-ku, Tokyo 101-8430, Japan}
 \affiliation{Department of Physics, Graduate School of Science, The University of Tokyo, 7-3-1 Hongo, Bunkyo-ku, Tokyo 113-0033, Japan}
\author{Mio Murao}
 \email{murao@phys.s.u-tokyo.ac.jp}
 \orcid{0000-0001-7861-1774}
 \affiliation{Department of Physics, Graduate School of Science, The University of Tokyo, 7-3-1 Hongo, Bunkyo-ku, Tokyo 113-0033, Japan}
 \affiliation{Trans-scale Quantum Science Institute, The University of Tokyo, 7-3-1 Hongo, Bunkyo-ku, Tokyo 113-0033, Japan}

\maketitle

\begin{abstract}
The quantum switch is a quantum process that creates a coherent control between different unitary operations, which is often described as a quantum process which transforms a pair of unitary operations $(U_1 , U_2)$ into a controlled unitary operation that coherently applies them in different orders as $\proj{0} \otimes U_1 U_2 + \proj{1} \otimes U_2 U_1$.
This description, however, does not directly define its action on non-unitary operations.
The action of the quantum switch on non-unitary operations is then chosen to be a ``natural'' extension of its action on unitary operations.  In general, the action of a process on non-unitary operations is not uniquely determined by its action on unitary operations.  It may be that there could be a set of inequivalent extensions of the quantum switch for non-unitary operations.
We prove, however, that the natural extension is the only possibility for the quantum switch for the 2-slot case.
In other words, contrary to the general case, the action of the quantum switch on non-unitary operations (as a linear and completely CP preserving supermap) is completely determined by its action on unitary operations.
We also discuss the general problem of when the complete description of a quantum process is uniquely determined by its action on unitary operations and identify a set of single-slot processes which are completely defined by their action on unitary operations.
\end{abstract}

\section{Introduction}
A quantum operation given by some physical process would have a predetermined action on a quantum system.
The action can be effectively converted by incorporating the quantum operation as a part of some sequence of quantum operations.
The sequence then serves as a ``higher-order quantum operation'' converting one quantum operation to another.
Quantum instruments~\cite{ozawa83} are a mathematical description of all quantum operations implementable by a fixed sequence of quantum gates within the standard quantum circuit model, including the addition and removal of subsystems, but without postselection.
Probabilistically implementable quantum operations are represented by completely positive (CP) and trace-nonincreasing maps.
Mathematically, a higher-order quantum operation on these quantum operations is represented by a higher-order map\footnote{By a ``map'', we mean a mathematically well-defined function, taking one element of some set to another element of some set.}.
We shall call any higher-order map on CP maps as a \textit{supermap} in this paper.
While a constructive definition of a supermap based on the standard quantum circuit model guarantees physical implementability (at least for theorists), an axiomatic characterization allows us to explore more exotic classes of (possibly unimplementable) supermaps.

Quantum processes\footnote{Some literatures use ``process matrix'' and ``superchannel'' to refer to higher-order maps.} are such supermaps, defined to be linear and return a CPTP map when the input maps are restricted to independent and local\footnote{These CPTP maps may act on operators that belong to a larger Hilbert space than that used to define the quantum processes, so long as the Hilbert space is ``locally'' extended from the defining Hilbert spaces of the quantum processes.} CPTP maps~\cite{indefinite_correlation,indefinite_purification}. Some quantum processes are not implementable within the standard quantum circuit and deemed to have an ``indefinite causal order''.
The causal orders of quantum processes, however, are not ``equally'' indefinite.
The quantum switch~\cite{qswitch} is a ``two-slot'' quantum process that transforms a pair of independent unitary operations into a controlled coherent superposition of these operations in two different orders. 
The quantum switch is well-studied in the context of indefinite causal order~\cite{qswitch,indefinite_correlation,indefinite_purification} to elucidate how the causal structure may affect the quantum information processing in quantum mechanics.
If implemented suitably, the quantum switch would serve as a resource for discriminating non-signaling channels~\cite{qs_discrimination1}, quantum computation~\cite{qs_discrimination2}, and quantum communication ~\cite{qs_comm_tri,qs_comm1,qs_comm2,qs_comm3,qs_comm_exp1,qs_comm_exp2}.

Reference \cite{qswitch} is the first to explicitly introduce the quantum switch, with two definitions.  The first one is valid only when the input CPTP maps are two independent local unitary CPTP maps.
Given two local unitary CPTP maps $\tmap{U}_1$ and $\tmap{U}_2$ characterized by the unitary operators $U_1$ and $U_2$\footnote{
We denote maps defined on linear operators with a tilde.},
the quantum switch transforms them to the unitary CPTP map given by
\begin{align}
(U_1 , U_2) \mapsto \proj{0}_C \otimes U_2 U_1 + \proj{1}_C \otimes U_1 U_2, \label{eq:def_qs_u}
\end{align}
by using each of $\tmap{U}_1$ and $\tmap{U}_2$ only once.  The input maps are applied in two different orders, determined by the state of another qubit (control qubit).
If the control qubit is in $\ket{0}_C$, $\tmap{U}_1$ is applied on the target system first, and followed by $\tmap{U}_2$,
and if the control qubit is in $\ket{1}_C$, $\tmap{U}_2$ is applied first followed by $\tmap{U}_1$.
If the control qubit is in a superposition of $\ket{0}$ and $\ket{1}$, then the two different orders should be ``superposed'' in a coherent way.

The second definition of the quantum switch includes all independent pairs of local CPTP maps as input and chosen to be the most ``natural'' one based on its Kraus representation~\cite{kraus} as follows:
given two CPTP maps represented by the Kraus operators $\{ K_i \}_i $ and $\{ L_j \}_j $, respectively,
the CPTP map of the output from the quantum switch is given by the Kraus operators $\{ W_{ij} \}_{ij}$ defined by
\begin{align}
W_{ij} := \proj{0}_C \otimes L_j K_i + \proj{1}_C \otimes K_i L_j . \label{eq:def_qs_k}
\end{align}
Following this definition\footnote{The quantum switch has been extended to larger numbers of input CPTP maps as in Ref.~\cite{qs_discrimination2}, but is not the scope of this paper.}, different results have been established.
One interesting result is that the quantum switch can help in enhancing certain communication tasks~\cite{qs_comm1,qs_comm2,qs_comm3}.
In particular, if two input CPTP maps are the depolarizing channels\footnote{Here, ``channel'' and ``CPTP map'' are used synonymously, but the former is chosen over the latter to conform to the convention commonly used in the community.}, which do not transfer any information by themselves alone,
the output CPTP map is not a depolarizing channel any more and can be used for transferring information.
The definition of the quantum switch based on Eq.~\eqref{eq:def_qs_k} is widely used in various studies, but before this paper, it was not clear if the action of the quantum switch based on the second definition can be derived uniquely from only Eq.~\eqref{eq:def_qs_u}.

Before proceeding, we present a simple example to show that, in general, the behaviour of a supermap on unitary CPTP maps does not uniquely define its action on general CPTP maps. Our example is illustrated in terms of quantum circuits shown in Fig.~\ref{fig:dep_equiv} and consists in presenting two quantum processes which, despite being different, have identical action on unitary CPTP maps. This example follows from the fact that since unitary CPTP maps are unital, that is,  $\tmap{U}(I)=I$.
If $\tmap{D}(\rho): = (\Tr\, \rho) \cdot I /d $ is the completely depolarizing channel and $\tmapf{id}(\rho) = \rho$ is the identity CPTP map, we have that
	\begin{align} \label{eq:unital}
	\tmap{D}\circ \tmap{U} \circ \tmap{D} 	 &=\tmapf{id}\circ \tmap{U} \circ \tmap{D}  = \tmap{D}.
	\end{align}
But, when considering non-unital CPTP maps $\tmapf{\Lambda}$, that is $\tmapf{\Lambda}(I)\neq I$, the identity in Eq.~\,\eqref{eq:unital} does not hold. In particular, if we take the CPTP map $\tmapf{\Lambda}(\rho) = (\Tr\, \rho)\proj{0}$ which replaces the input state $\rho$ by the quantum state $\proj{0}$, we obtain
\begin{align} \label{eq:non_unital}
	&\tmap{D}\circ \tmapf{\Lambda} \circ \tmap{D} =	 \tmap{D} \\
	& \tmapf{id}\circ \tmapf{\Lambda} \circ \tmap{D} = \tmapf{\Lambda}.
	\end{align}

\begin{figure}
\texttt{\quad $\forall$ all unitary operations $\tmap{U}$:}
\includegraphics[width=\hsize]{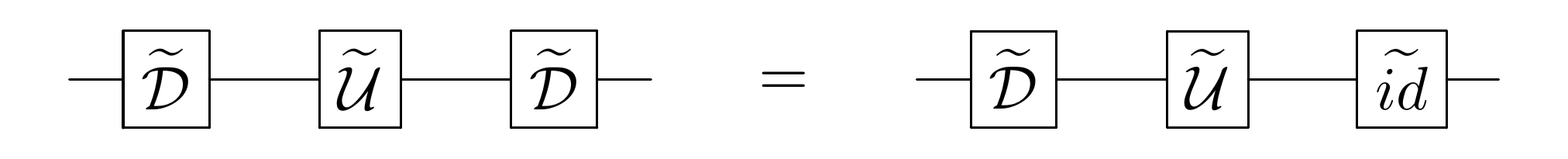}
\texttt{but $\exists$ non-unitary operations such as $\tmapf{\Lambda}$:}
\includegraphics[width=\hsize]{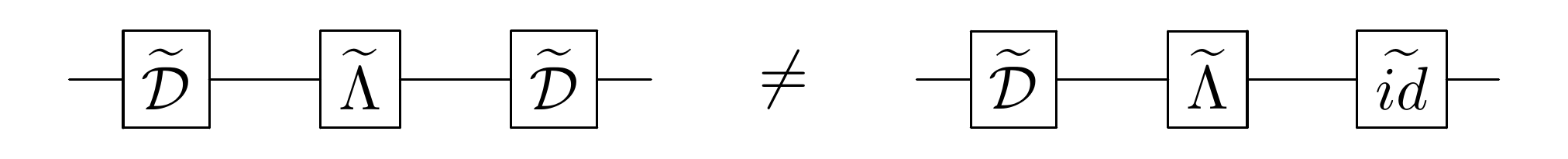}
\caption{Two quantum circuits implementing higher-order maps that transform an arbitrary unitary CPTP map into the depolarizing channel.
Here $\tmap{U}$ denotes the input unitary CPTP map, $\tmap{D}$ denotes the depolarizing channel, and $\tmapf{id}$ denotes the identity CPTP map.
While the actions of both of them are equivalent on unitary CPTP maps, their actions on some CPTP maps,
e.g., $\tmapf{\Lambda}(\rho) = (\Tr\, \rho) \proj{0}$ which replaces the input state $\rho$ by the quantum state $\proj{0}$, are different.
}\label{fig:dep_equiv}
\end{figure}
Although simple, this example can be used to construct infinitely many pairs of higher-order quantum operations that share this property, namely, the action on unitary operations does not uniquely define the action on non-unitary operations. In particular, for any fixed circuit operation, we can construct two new circuits by considering probabilistic mixtures of the fixed operations with the identity channel or the depolarizing channel. For any non-trivial probability mixture, the same argument used in our previous example shows that these two new circuits may act identically on unitary operations, but they act differently for non-unital operations. 
More rigorously, let $C$ be an arbitrary one-slot quantum comb \cite{comb1,comb2,comb3} representing a circuit which a channel into a channel. Let $C_1$ be the comb associated to the circuit that transforms an arbitrary channel $\tmapf{\Lambda}$ into $\tmap{D}\circ \tmapf{\Lambda} \circ \tmap{D} =	 \tmap{D}$ and $C_2$ be the comb associated to the circuit transforms an arbitrary channel $\tmapf{\Lambda}$ into $\tmapf{id}\circ \tmapf{\Lambda} \circ \tmap{D} = \tmapf{\Lambda}$. It follows from Eq.~\eqref{eq:unital} and Eq.~\eqref{eq:non_unital} that for any $0 \leq \eta<1$ the circuits associated to the combs
\begin{align}
    C^\eta_1:=\eta \,C + (1-\eta) C_1 \\
    C^\eta_2:=\eta \,C + (1-\eta) C_2,
\end{align}
are equivalent for unitary operations, but $C^\eta_1$ and $C^\eta_2$ act differently on non-unitary quantum channels.
Interestingly, this argument does not rely on quantum circuits and can be used for any general process (potentially with indefinite causal order~\cite{qswitch}) which transforms quantum operations. We just need to recognize that the two different circuits presented in Fig.~\ref{fig:dep_equiv} are valid processes, and by performing any non-trivial convex combinations, one can always construct examples of two higher-order transformations in which the action on unitary operations does not uniquely define the action on general.

It is worth mentioning that the difficulty of the problem follows from the fact that general CPTP maps cannot be written as a linear combination of unitary ones. For instance, a non-unital CPTP map can never be written as a linear combination of unitary CPTP maps. However, the linearity is not the only condition that supermaps should respect. It is crucial to respect the notion of complete positivity (CP)~\cite{comb1,comb2,comb3,qs_comm1,indefinite_correlation,indefinite_purification}.

The question of whether supermaps are uniquely defined by their action on unitary CPTP maps is first considered in Ref.~\cite{qswitch}, where the authors have analysed the incoherent version of the quantum switch---a bilinear supermap that transforms an independent pair of local CPTP maps into a classically-controlled mixture of the pair applied in two possible orders.
Also in Ref.~\cite{qswitch}, the authors have shown that the incoherent switch is uniquely defined by its action on unitary CPTP maps. Here, we prove that the coherent quantum switch is also uniquely defined by its action on unitary CPTP maps under linearity and completely CP preserving (see definition below).

When dealing with the ``lower-order case'' where one considers CPTP maps which transform quantum states (or, rather, density operators) into quantum states, the linearity itself is enough to ensure that the action of the maps on pure quantum states uniquely defines the action of the maps on any quantum state\footnote{The same applies to CP instruments.}. This holds because the linear space spanned by pure quantum states contains all self-adjoint operators.
An analogous problem then would be to define the action of a CPTP map for a restricted set of pure states and then analyse how the CP condition restricts the extension of this map to the set of all quantum states in a non-trivial way, a problem which is discussed in Refs.~\cite{chefles00,chefles03,extending,ncopy}. We also remark Ref.~\cite{haapasalo14} which investigates when the marginals of a linear map uniquely determine its action and Ref.~\cite{guerini18} which analyses when compatible tuples of measurements admit a unique joint measurement.

In this paper, we analyse how the action of supermaps on unitary CPTP maps restricts their action on general CPTP maps. In particular, we show that the constraints of linearity and completely CP preserving ensure that the action of the quantum switch on any CP map (including non-TP ones) is uniquely defined by its action on unitary CPTP maps.  We then consider the general problem of understanding the restrictions imposed by the completely CP preserving and apply our proof techniques to problems which are not directly related to the quantum switch. In particular, we identify a class of single-slot supermaps (which take one CPTP map as input) such that, similarly to the quantum switch, their action on unitary CPTP maps uniquely defines their action on general CP maps. This class includes identity supermaps, unitary transposition, and qubit unitary complex conjugation.

This paper is organized as follows.
In Sec.~\ref{sec:qs}, we provide the statement of our main theorem.
Section~\ref{sec:span} discusses several consequences that can be obtained only from the linearity that are also used in later proofs.
The results on the single-slot supermaps are derived in Sec.~\ref{sec:single}.
Finally, we conclude in Sec.~\ref{sec:conclusion}. The proofs and technical methods are presented in the Appendix.

\section{The Uniqueness of Quantum Switch}\label{sec:qs}
Linearity and completely CP preserving (to be defined precisely, below) are two operationally motivated conditions on supermaps.
The linearity condition is required for the supermap to accommodate the possibility of probabilistic mixtures in quantum theory~\cite{barrett07}.
A supermap must also be completely CP preserving so that it always returns a CP map when the input CP map acts on extended Hilbert spaces~\cite{comb2,comb3}.  This is to comply with the possibility that any quantum system can be extended by adding auxiliary systems.

In addition to the complete CP preserving, a supermap may be imposed a set of normalization constraints, namely that a pair of quantum channels transform into a quantum channel~\cite{qswitch}. These constraints imply that all supermaps lie in some particular linear subspace. Interestingly, we do not impose such constraints for proving our result. These normalization constraints will follow as consequences of the linearity, completely CP preserving, and Eq.~\eqref{eq:def_qs_u}.

In this paper, we start from a two-slot supermap whose action on unitary operations matches that of the quantum switch and derive the action of the supermap on arbitrary quantum operations.
In particular,
we derive the quantum switch from the following three conditions:
 \begin{enumerate}
 \item The action of a two-slot supermap on unitary operations is given by Eq.~\eqref{eq:def_qs_u}.
 \item The supermap is linear.
 \item The supermap is a completely CP preserving supermap.
 \end{enumerate}

\begin{thm}\label{thm:qs_unique}
If a 2-slot supermap is linear and completely CP preserving and its
action on any independent pair of local unitary CPTP maps is given by Eq.~\eqref{eq:def_qs_u},
then its action on arbitrary independent pair of local CPTP maps is uniquely determined by the Kraus operators Eq.~\eqref{eq:def_qs_k}.
\end{thm}

\begin{figure}
\centering
\includegraphics[width=0.6\hsize]{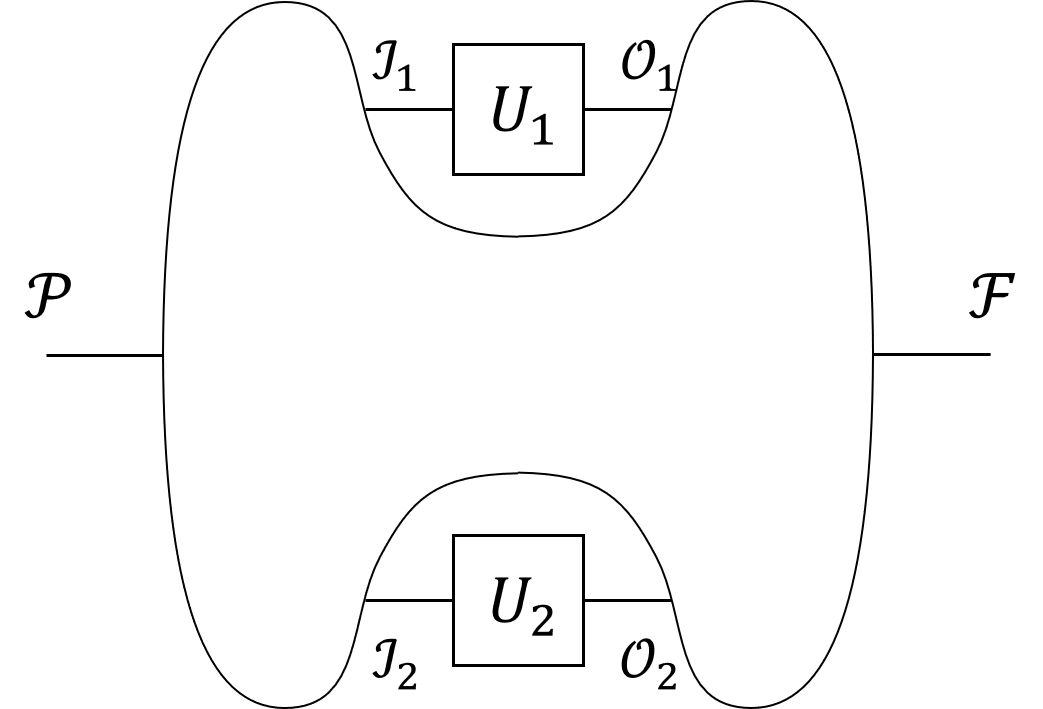}
\caption{Conceptual figure of a two-slot supermap which may not respect a definite causal order, such as the quantum switch. Two-slot supermaps transform an independent pair of local CPTP maps as input and output a CPTP map from $\pcal$ to $ \fcal$. In general, the action of a supermap on unitary CPTP maps does not uniquely define its action on non-unitary CPTP maps.
}\label{fig:qs}
\end{figure}

In order to prove Theorem~\ref{thm:qs_unique}, we use the Choi representation of supermaps~\cite{comb1,comb2,choi,jamiolkowski}.
That is, since the Choi-{\Jamiolkowski} isomorphism is a bijection,
we prove the uniqueness of the quantum switch by proving that the corresponding Choi operator $W$ is uniquely determined.
Representing supermaps~\cite{comb2} as matrices via such an isomorphism is the same approach taken in studying “process matrices” as in Ref.~\cite{Oreshkov_2016}.

We now introduce the Choi representation of maps on linear operators ~\cite{choi,jamiolkowski}. Then we discuss how to represent supermaps as linear operators via the Choi representation~\cite{comb1,comb3}. Let $\mathcal{I}$ and $\mathcal{O}$ be a Hilbert space, and $\mathcal{L}(\mathcal{I})$ and $\mathcal{L}(\mathcal{O})$ the set of linear operators acting on them. The Choi operator of a map $\tmapf{\Lambda}: \lcal(\ical) \to \lcal(\ocal)$ transforming operators from $\mathcal{L}(\mathcal{I})$ to $\mathcal{L}(\mathcal{O})$ is given by
\begin{align}
J_{\tilde{\Lambda}} := \sum_{ij} \ketbra{i}{j} \otimes \tmapf{\Lambda}( \ketbra{i}{j}) \in \lcal(\ical \otimes \ocal),
\end{align}
where $\{ \ket{i} \}$ is a computational basis for $\ical$.
The condition that a linear map $\tmapf{\Lambda}$ is completely positive (CP) corresponds to the positivity of $J_{\tilde{\Lambda}}$, i.e., $J_{\tilde{\Lambda}} \geq 0$.
Unitary CPTP maps are the ones which can be written as $\tmap{U}(\rho)=U\rho U^\dagger$, where $U$ is a unitary operator.
The Choi operator associated to a unitary CPTP map will have a subscript in a Roman typeface, for example, $J_U := \sum_{ij} \ketbra{i}{j} \otimes (U \ketbra{i}{j}U^\dagger)$.

A supermap is a map of
\begin{equation} \label{eq:smap}
  \ttsmap{W}:  [\lcal(\ical) \to  \lcal(\ocal)]  \to
	  [\lcal(\pcal) \to  \lcal(\fcal)],
\end{equation}
where $\pcal$ and $\fcal$ are the input and output Hilbert spaces of the output CP map, respectively.
$\ttsmap{W}$ is linear if
\begin{equation}
  \ttsmap{W} (\alpha \tmap{A} + \beta \tmap{B}) = \alpha \ttsmap{W} ( \tmap{A}) + \beta \ttsmap{W} ( \tmap{B} )
\end{equation}
for arbitrary input CP maps $\tmap{A}$ and $\tmap{B}$ and arbitrary $\alpha, \beta \in \mathbb{C}$.
It is CP preserving if it transforms CP maps into CP maps.

A linear supermap is completely CP preserving if all of its trivial extensions transform CP maps into CP maps, that is, for every CP map which may make use of an auxiliary space $\tmap{C}:\lcal(\ical\otimes \mathcal{H}_\text{aux}) \to  \lcal(\ocal\otimes\mathcal{H}_\text{aux})$,  the output map $\ttsmap{W} \otimes \smap{id}( \tmap{C} )$
is also CP, where $\smap{id}:\lcal(\mathcal{H}_\text{aux}) \to  \lcal(\mathcal{H}_\text{aux})$ is the identity supermap on the auxiliary space defined by $\smap{id}(\tmap{A})=\tmap{A}$ for every map $\tmap{A}$.

Supermaps may also be viewed as maps on the Choi operators of the input CP maps. Hence, by taking the Choi representation of this induced map, we can also represent a supermap as in Eq.~\,\eqref{eq:smap} by an operator ${W\in \lcal(\pcal \otimes \ical \otimes \ocal \otimes \fcal)}$.
More precisely, if a supermap $\ttsmap{W}$ acts on $\tmap{A}$, the Choi operator of its output CPTP map $\ttsmap{W}(\tmap{A})$ is given by~\cite{comb1,comb3,qswitch}
\begin{equation}
\Tr_{\ical\ocal} \Big[W (I_\pcal \otimes A_{\ical\ocal} \otimes I_\fcal)^t\Big],
\end{equation}
where $A_{\ical\ocal} \in\lcal(\ical \otimes \ocal)$ is the Choi operator of $\tmap{A}$ and $t$ stands for transposition in the computational basis.
In the Choi operator representation, the completely CP preserving condition of $\ttsmap{W}$ is simply given by the positivity constraint $W \geq 0$~\cite{comb1,comb3}.
A more detailed introduction to the Choi operator representations of supermaps can be found, for instance, in the introductory sections of Ref.~\cite{indefinite_reversible} and~\cite{unitary_inversion2}, also references therein.

The quantum switch is a 2-slot supermap, namely that the $\ical$ and $\ocal$ are each composed of two Hilbert spaces as $\ical = \ical_1 \otimes \ical_2$ and $\ocal = \ocal_1 \otimes \ocal_2$ (see Fig.\,\ref{fig:qs} for assignments of the Hilbert spaces).
Here, $\ical_k, \ocal_k \cong \mathbb{C}^d$ for $k=1,2$ and $\pcal, \fcal \cong \mathbb{C}^{2d}$.
In particular, $\pcal = \pcal_C \otimes \pcal_T$ and $\fcal = \fcal_C \otimes \fcal_T$, where $\pcal_T, \fcal_T \cong \mathbb{C}^d$ and $\pcal_C, \fcal_C \cong \mathbb{C}^2$, the latter corresponding to the control qubit.
We denote the Choi operator of the quantum switch defined on any independent pair of local CPTP maps $\tmap{A} \otimes \tmap{B}$, where $\tmap{A} : \lcal(\ical_1) \to \lcal(\ocal_1)$ and $\tmap{B} : \lcal(\ical_2) \to \lcal(\ocal_2)$, by the Kraus operators given in Eq.~\eqref{eq:def_qs_k} as
 $W_0 \in \lcal(\pcal \otimes \ical_1 \otimes  \ocal_1\otimes \ical_2\otimes \ocal_2\otimes \fcal)$.
This operator can be calculated~\cite{qswitch} as $W_0 := \proj{W_0}$ with
\begin{align}
\ket{W_0} &:= \ket{00}_{\pcal_C \fcal_C} \kket{I}_{\pcal_T \ical_1} \kket{I}_{\ocal_1 \ical_2} \kket{I}_{\ocal_2 \fcal_T} \notag \\
&\quad+ \ket{11}_{\pcal_C \fcal_C} \kket{I}_{\pcal_T \ical_2} \kket{I}_{\ocal_2 \ical_1} \kket{I}_{\ocal_1 \fcal_T} \notag \\
&= \sum_{i,j,k=0}^{d-1} ( \ket{ijjkik00} + \ket{jkijik11} ), \label{eq:def_qs_w0}
\end{align}
where the Hilbert spaces are rearranged for ease of presentation, $\kket{I}:=\sum_{i}\ket{ii}$, and the Hilbert spaces for the vector in the last line is in the order of $\ical_1 \otimes \ocal_1 \otimes \ical_2 \otimes \ocal_2 \otimes \pcal_T  \otimes \fcal_T \otimes \pcal_C \otimes \fcal_C$.
Note that the normalization of the Choi operators is given by $\Tr [W] = \Tr [W_0] = 2 d^3$, and each element of the matrix $W_0$ is 0 or 1.

Denoting the input unitary CPTP map by $J = J_{U_1} \otimes J_{U_2}$, the action of quantum switch is given by\footnote{When we write  $\Tr_{in} [W_0 J^t]$, we are suppressing the identity operators, that is, in the full notation we have  $\Tr_{in}\big[ W_0 (I_\pcal\otimes J^t\otimes I_\fcal)\big]$.  If clear from context, we will suppress the identity operators under traces.}
 $\Tr_{in} [W_0 J^t]$, where $\mathcal{H}_{in}:=\ical_1 \otimes  \ocal_1\otimes \ical_2\otimes \ocal_2$ stands for the total input-channel space of the quantum switch.
In terms of the Choi representation, Theorem~\ref{thm:qs_unique} is equivalent to the following Lemma.
\begin{lem}\label{lem:qs_unique}
If the Choi operator $W$ of a 2-slot linear supermap satisfies $W \geq 0$ (i.e., completely CP preserving) and $\Tr_{in} [ W (J_{U_1} \otimes J_{U_2})^t ] = \Tr_{in} [ W_0 (J_{U_1} \otimes J_{U_2})^t] $ for all $U_1$ and $U_2$, then $W = W_0$.
\end{lem}
The main proof of Lemma~\ref{lem:qs_unique} is given in Appendix~\ref{ap:proof_qs}.
Here we provide a sketch of the proof.
We first consider the linear span of the Choi operator of unitary CPTP maps.
In Sec.~\ref{sec:span} and Appendix~\ref{ap:span}, we provide a set of operators that are in the linear span.
In Appendix~\ref{ap:proof_qs}, we consider the consequences of the positivity of the Choi operator $W \geq 0$, which also relies on the results on the linear span.
Remark that the Choi operator of the quantum switch given by Eq.~\eqref{eq:def_qs_w0} has a simple form, namely,
it is a matrix such that all the elements are either 1 or 0 and that any submatrix whose diagonal elements are 1 has all other elements being 1.
First we prove that all the diagonal elements of $W$ are either 1 or 0, with some calculations are separated into Appendix~\ref{ap:diag}.
After determining all the diagonal elements, it remains to show that all other elements are either 0 or 1.
To do so, we evaluate the element-wise 1-norm, i.e., the sum of absolute values of all elements in the computational basis, where the main calculations are separated into Appendix~\ref{ap:offdiag}.

\section{The Linear Span of the Choi Operators of Unitary Operations}\label{sec:span}

From the linearity, the action on any input CP map that is a linear combination of unitary CPTP maps follows from Eq.~\eqref{eq:def_qs_u}.
For any $J \in \mathrm{span} \{  J_{U_1} \otimes J_{U_2}  \} $\footnote{Here, $J_{U_1} \otimes J_{U_2}$ is treated as a collection of all the Choi operators of local unitary CPTP maps.  We use this as a shorthand notation of $\mathrm{span} \{ \{ J_{U_1} \otimes J_{U_2} \}_{U_1,U_2 \in SU(d)} \}$.}, where
\begin{align}
&\mathrm{span} \{ J_{U_1} \otimes J_{U_2} \} \\
&\quad:= \{ C \mid C = \sum_{i_1, i_2} c_{i_1 i_2} (J_{U_{i_1}} \otimes J_{U_{i_2}} ), c_{i_1 i_2} \in \mathbb{C}, \notag\\
&\quad\quad\quad U_{i_1},U_{i_2} \in SU(d) \notag  \},
\end{align}
the action of the quantum switch on $J$ is given by $\Tr_{in} [W_0 J^t]$.
Since $\mathrm{span} \{ J_{U_1} \otimes J_{U_2} \} = \mathrm{span} \{ \mathrm{span} \{ J_U \} \otimes \mathrm{span} \{ J_U \} \}$,
tensor products of the bases of $\mathrm{span} \{ J_U \}$ form a linearly independent set of elements in $\mathrm{span} \{ J_{U_1} \otimes J_{U_2} \}$.

We first summarize the results in this section:
as we prove in Appendix~\ref{ap:span}, an operator $C$ belongs to $\mathrm{span} \{ J_U \}$ if and only if
\begin{align} \label{Eqcondition}
	\Tr_{\mathcal{I}}[C]=& \Tr[C] \, \frac{I_d}{d}, \\
	\Tr_{\mathcal{O}}[C]=& \Tr[C] \, \frac{I_d}{d}.
\end{align}
Hence, by evaluating partial traces, one can verify that:
\begin{enumerate}
\item $\ketbra{ii}{jj}, \ketbra{ij}{ji}  \in \mathrm{span} \{ J_U \}$ for  $i \neq j$.
\item $\ketbra{ij}{i'j'} \in \mathrm{span} \{ J_U \}$ when the indices $i,j,i',j'$ all take a different value.
\item $\ketbra{ij}{i'j'} \in \mathrm{span} \{ J_U \}$ when the indices $i,j,i',j'$ take three distinct values, but such that their partial traces all return a zero operator.
\item $\sum_{j=i+k} \ketbra{ij}{ij} \in \mathrm{span} \{ J_U \}$ for $k=0,\ldots,d-1$.
\item $( \ketbra{ij}{ii} - \ketbra{jj}{ji} ) , ( \ketbra{ji}{ii} - \ketbra{jj}{ij} ) \in \mathrm{span} \{ J_U \}$ for $i \neq j$.
\end{enumerate}

We remark that the set of the operators listed here does not span the whole space of $\mathrm{span} \{ J_U \}$ unless $d=2$.
There are only $2d(d-1) + d(d-1)(d-2)(d-3) + 4d(d-1)(d-2) + d + 2d(d-1) = d (d^3 -2 d^2  + 3 d  -1)$ elements here,
but the dimension of $\mathrm{span} \{ J_U \}$ is $(d^2-1)^2 + 1$ (Lemma~3 of Ref.~\cite{sn}), which is strictly greater than $d (d^3 -2 d^2  + 3 d  -1)$ for $d > 2$.
We show in Appendix~\ref{ap:proof_qs} that Lemma~\ref{lem:qs_unique} can be proved by only considering the action of the quantum switch on these operators.
The action of the quantum switch on all unitary CPTP maps is not necessary to uniquely determine the action of the quantum switch on general CP maps.

Another important fact is that the action of the quantum switch on the depolarizing channels is determined only by the linearity,
because
\begin{align}
\frac{I_{d^2}}{d} \otimes \frac{I_{d^2}}{d} \in \mathrm{span} \{ J_{U_1} \otimes J_{U_2} \}
\end{align}
holds, where $\frac{I_{d^2}}{d}$ is the Choi operator of the depolarizing channel (which may be viewed as a probabilistic mixture of uniformly random unitary operations).
Thus, the results of Refs.~\cite{qs_comm1,qs_comm2,qs_comm3} can be obtained with only the assumption of the linearity.

\section{Uniqueness of single input supermaps}\label{sec:single}

In this section, we provide a similar theorem on the supermaps with a single input CP map instead of the quantum switch.
Consider the (linear) supermap $\ttsmap{C}$ which act on as the identity supermap on unitary CPTP maps, that is, for every
$\tmap{U}$ we have $\ttsmap{C}(\tmap{U})=\tmap{U}$.
We show that by imposing that  $\ttsmap{C}$  is completely CP preserving, $\ttsmap{C}$ also behaves as the identity supermap for any arbitrary CP map $\tmapf{\Lambda}$, that is, $\ttsmap{C}(\tmapf{\Lambda})=\tmapf{\Lambda}$.
In other words, we have that $C = C_0$, where
\begin{align}
C_0 := \sum_{i,j,k,l=1}^d \ketbra{ij}{kl} \otimes \ketbra{ij}{kl} \in \lcal( \ical \otimes \ocal \otimes \pcal \otimes \fcal), \label{eq:def_id_supermap}
\end{align}
the Choi operator of the identity supermap.
\begin{thm}\label{thm:id_unique}
If the action of a 1-slot linear supermap on unitary CPTP maps is given by $\tmap{U} \mapsto \tmap{U}$ for an arbitrary input unitary operation $\tmap{U}$ and is completely CP preserving,
then its action on an arbitrary CP map $\tmapf{\Lambda}$ is given by $\tmapf{\Lambda} \mapsto \tmapf{\Lambda}$.
\end{thm}
Similarly, Theorem~\ref{thm:id_unique} is equivalent to the following Lemma.
\begin{lem}\label{lem:id_unique}
Let $C_0$ be the Choi operator of the identity supermap defined by Eq.~\eqref{eq:def_id_supermap}.
If there exists $C \geq 0$ such that $\Tr_{\ical \ocal} [ C J_{U}^t ] = \Tr_{\ical \ocal} [ C_0 J_{U}^t ] $ for all ${U} $,
then $C = C_0$.
\end{lem}
\noindent The proof of Lemma~\ref{lem:id_unique} is given in Appendix~\ref{ap:proof_id}.

While Theorem~\ref{thm:id_unique} only holds for supermaps which act as the identity on unitary CPTP maps, this
result can be easily extended to cases to supermaps which act on unitary CPTP maps as $\tmap{U} \mapsto \tmap{B} \circ \tmap{U} \circ \tmap{A}$, where $\tmap{A}$ and $\tmap{B}$ are fixed unitary CPTP maps that do not make use of any auxiliary system. This class of supermaps is precisely the set of reversibility preserving one-slot superchannels (also known as one-partite pure processes~\cite{indefinite_purification}) which do not modify the dimension of the input CP map~\cite{indefinite_reversible}.

\begin{coll}\label{coll:single_unique}
If the action of a 1-slot supermap $\ttsmap{C'}$ on unitary CPTP maps is given by $\tmap{U} \mapsto \tmap{B} \circ \tmap{U} \circ \tmap{A}$ for an arbitrary input unitary CPTP map $\tmap{U}$ and two fixed unitary CPTP maps $\tmap{A}$ and $\tmap{B}$ and is completely CP preserving,
then its action on arbitrary CP maps is given by $\tmapf{\Lambda} \mapsto \tmap{B} \circ \tmapf{\Lambda} \circ \tmap{A}$ for an arbitrary input CP map $\tmapf{\Lambda}$.
\end{coll}
\begin{proof}
Corollary~\ref{coll:single_unique} can be reduced to Theorem~\ref{thm:id_unique} as follows.
In the Choi operators, the action of $\ttsmap{C'}$ on unitary CPTP maps is
\begin{align}
\Tr_{\ical \ocal} [ C' J_{U}^t ] = J_{BUA},
\end{align}
where $A$ and $B$ are unitary operators defining $\tmap{A}$ and $\tmap{B}$, respectively.
By introducing $C = (A_{\ical} \otimes B_{\fcal})^\dag C' (A_{\ical} \otimes B_{\fcal})$,
the l.h.s. can be written as
\begin{align}
\Tr_{\ical \ocal} [ C' J_{U}^t ]
&= \Tr_{\ical \ocal} [(A_\ical \otimes B_\fcal) C (A_\ical \otimes B_\fcal)^\dag J_U^t] \notag \\
&= B_\fcal \{ \Tr_{\ical \ocal} [C ( A_\ical^\dag J_U^t A_\ical )] \} B_\fcal^\dag \notag \\
&= B_\fcal \{ \Tr_{\ical \ocal} [C ( J_{AU} )^t ] \} B_\fcal^\dag.
\end{align}
We also have that
\begin{align}
J_{BUA} = B_\fcal \{ \Tr_{\ical \ocal} [C_0 ( J_{AU} )^t ] \} B_\fcal^\dag,
\end{align}
where $C_0$ is the operator defined in Eq.~\eqref{eq:def_id_supermap}.
Since $AU$ is also a unitary operator and spans the whole linear span of $\mathrm{span}\{ J_U \}$,
the uniqueness of $C'$ is equivalent to the uniqueness of the identity supermap $C$.
\end{proof}

Corollary~\ref{coll:single_unique} can be used to show that the supermap leading to qubit unitary complex conjugation analysed in Ref.~\cite{unitary_conjugate} is uniquely defined by its action on unitary operations. This holds true because
$\tmap{U}^* = \tmap{Y} \circ \tmap{U} \circ \tmap{Y}$, where $\tmap{Y}$ denotes the Pauli Y unitary operation.

Another important supermap that can be proven to be unique is the unitary transposition~\cite{unitary_inversion1,unitary_inversion2}. Here it is interesting to remark that unitary transposition cannot be implemented deterministically and its probability of success is limited by $1/d^2$. Note however that our proofs only rely on the linearity and the completely CP preserving and do not make use of normalization or causal constraints. Due to that, we can also tackle probabilistic higher-order quantum operations such as unitary transposition.

\begin{coll}[Unitary transposition]\label{coll:transpose_unique}
Suppose a 1-slot linear and completely CP preserving supermap transforms a unitary CPTP map determined by a unitary operator $U$ into another unitary CPTP map determined by its transpose operator $U^t$ with some fixed success probability $p$. More formally, the Choi operator $C\geq0$ of the supermap satisfies $\Tr_{\ical \ocal} [ C J_{U}^t ]=p J_{U^t}$ for all unitary operators $U$ and some $p$. The action of $C$ on a general CP map with the Choi operator $J_\Lambda\in\mathcal{L}(\mathcal{I}\otimes\mathcal{O})$ is necessarily given by $\Tr_{\ical \ocal} [ C J_\Lambda^t ] =p \, F  J_\Lambda F$, where $F$ denotes the flip (swap) operator defined by $F \ket{\psi \phi} = \ket{\phi \psi}$ for all $\ket{\phi}$ and $\ket{\psi}$.
\end{coll}
\begin{proof}
A possible Choi operator for unitary transposition is given by $C'_0 = F_{\ical \ocal} C_0 F_{\ical \ocal}^\dag$, where $C_0$ is the operator defined in Eq.~\eqref{eq:def_id_supermap}. By assumption,
\begin{align}
\Tr_{\ical \ocal} [C' J_{U}^t] = \Tr_{\ical \ocal}  [C'_0 J_{U}^t].
\end{align}
In this case, by introducing $C = (F_{\ical \ocal})^\dag C' (F_{\ical \ocal})$,
the l.h.s. can be written as
\begin{align}
\Tr_{\ical \ocal}  [C' J_{U}^t]
&= \Tr_{\ical \ocal} [(F_{\ical \ocal}) C (F_{\ical \ocal})^\dag J_U^t] \notag \\
&= \Tr_{\ical \ocal} [ C (F_{\ical \ocal})^\dag J_U^t (F_{\ical \ocal})] \notag \\
&= \Tr_{\ical \ocal} [ C (J_{U^t})^t]
\end{align}
whereas the r.h.s. can be written as
\begin{align}
\Tr_{\ical \ocal} [C'_0 J_{U}^t] = \Tr_{\ical \ocal} [ C (J_{U^t})^t].
\end{align}
Since $U^t$ is also a unitary operator and spans the whole linear span of $\mathrm{span}\{ J_U \}$,
the uniqueness of $C'$ is equivalent to the uniqueness of the identity supermap $C$.
\end{proof}
After finishing a first version of this manuscript, we became aware of Ref.~\cite{chiribella20}, which presents an independent and considerably different proof that the unitary transposition is uniquely defined by the action of on unitary CPTP maps.

\begin{rem}
In general, for a given action of a supermap on unitary CPTP maps,
an extension of the domain of the supermap is not necessarily unique under the linearity and the completely CP preserving.
Thus, it is natural to consider what extra condition may lead to the uniqueness.
The results of Corollary~\ref{coll:single_unique} shows  the Choi operator of the supermap $\tmap{U} \mapsto \tmap{B} \circ \tmap{U} \circ \tmap{A}$ is a rank-1 operator.
One could imagine that, when there exists a rank-1 positive semidefinite extension, that operator is the only possible extension.
Unfortunately, this hypothesis is not enough for the uniqueness of the extension, as we will now show.

Consider a supermap transforming all two-dimensional unitary CPTP maps
to the two-dimensional identity channel, with some probability.
That is, there exists a Choi operator $C_0$ such that $\Tr_{\ical \ocal} [ C_0 J_{U}^t ] =q_U J_I$, where $q_U$ is a positive coefficient.
For $0 \leq p \leq 1$, we define the Choi operator  $C_p \in \mathcal{L} \big(\ical \otimes \ocal \otimes \pcal \otimes \fcal\big)$
\begin{align}
C_p:=
\begin{pmatrix}
1 & p &  p & 1 \\
p & 1 & 1 & p \\
p & 1 & 1 & p \\
1 & p &  p & 1
\end{pmatrix}
\otimes \phi^+,
\end{align}
where $\phi^+$ denotes the maximally entangled state, which is proportional to the Choi operator of the identity channel.
We can check that, for any $0 \leq p \leq 1$, the positivity constraint $C_p\geq0$ is respected and $\Tr_{\ical \ocal} [C_p J_{U}^t] = q_U J_I$, that is, $C_p J_{U}^t $ is proportional to $J_I$, but the proportionality coefficient $q_U$ is independent of $p$ but may depend on the unitary $U$. However, when $p=1$, $C_p$ is a rank-1 operator. Hence, imposing the existence of a supermap with rank-1 Choi operator is not sufficient to uniquely extend the action of such a supermap on non-unitary CPTP maps.
\end{rem}

\section{Conclusion and Future Scope}\label{sec:conclusion}

In this paper, we have addressed whether the action of a supermap is uniquely determined by its action on unitary CPTP maps.
We proved that, different from general supermaps, the action of the quantum switch on unitary CPTP maps uniquely defines its action on general CP maps under linearity and completely CP preserving.
We have also shown this uniqueness property holds for different classes of supermaps, which include identity supermaps, unitary transpositions~\cite{unitary_inversion2}, and the qubit unitary complex conjugation~\cite{unitary_conjugate}.

In order to tackle more general supermaps, it would be useful to find a simple criterion to identify which supermaps are uniquely defined by their action on unitary CPTP maps. Here particular cases of interest would be to analyze the $n$-slot quantum switch~\cite{qs_discrimination2}, and quantum processes which may be viewed as generalizations of the switch, such as the ones which can be described by a coherent control of causal order~\cite{indefinite_reversible,barrett21,wechs2021quantum}. 
Also, the quantum switch has the property of being a pure superchannel~\cite{indefinite_purification}, that is, it maps unitary operations (even unitary operations extended to larger systems) into unitary operations, and preserve the reversibility of quantum operations~\cite{indefinite_reversible,barrett21}. It would be interesting to have a better understanding of the relationship between pure superchannels and uniqueness property, as in the case of quantum switch. For instance, is it the case that all pure superchannels are  uniquely defined among pure superchannels by this action on (non-extended) unitary operations? We stress however that in this work, we do not assume that the switch is a pure superchannel, the only properties that we assume are its action on (non-extended) unitary operations, linearity, and complete positiveness. From these hypotheses, we prove that this has to be a pure superchannel, and, additionally, that it is unique.

Another interesting future direction would be to understand how the constraints of normalization would limit the possible extensions of supermaps. In this work we have obtained our results simply by arguments based on completely CP preserving, but the normalization constraints for deterministic transformations or linear constraints arising from causal relations~\cite{comb3} may considerably restrict the possibilities of extending the action of supermaps to non-unitary CPTP maps. We note that, when extending lower-order operations from a restricted set of quantum states, example 5 of Ref.~\cite{extending} shows that requiring the transformation to be deterministic may be a severe constraint. We then expect that stronger claims may be proven by imposing that a supermap correspond to a deterministic transformation.

As we show here, it is possible to strongly restrict the form of a supermap given only its action on unitary CPTP maps, and such a technique may reduce the difficulty of the analysis of the performance of higher-order quantum operations, since it is enough to perform an optimization on a restricted set of CP maps.
We hope that questions, answers, and methods covered in this paper contribute to our understanding of how positivity conditions restrict higher-order quantum operations and that this work paves the way for understanding the extension of general supermaps.

\section*{Acknowledgments}

We acknowledge Simon Milz for comments on the first version of this manuscript.
This work was supported by MEXT Quantum Leap Flagship Program (MEXT QLEAP) Grant No.~JPMXS0118069605, JPMXS0120351339,
Japan Society for the Promotion of Science (JSPS) KAKENHI Grant No.~17H01694, 18H04286, 18K13467, 21H03394
and Advanced Leading Graduate Course for Photon Science (ALPS).
M.~T.~Quintino acknowledges the Austrian Science Fund (FWF) through the SFB project ``BeyondC'' (sub-project F7103),
a grant from the Foundational Questions Institute (FQXi) Fund and a grant from the John Templeton Foundation (sub-project F7103)
as part of the Quantum Information Structure of Spacetime (QISS) Project (qiss.fr).
The opinions expressed in this publication are those of the authors and do not necessarily reflect the views of the John Templeton Foundation.
This project has received funding from the European Union’s Horizon 2020 research and innovation programme under the Marie Sk\l odowska-Curie grant agreement No 801110 and the Austrian Federal Ministry of Education, Science and Research (BMBWF).
It reflects only the authors' view, the EU Agency is not responsible for any use that may be made of the information it contains. A. Soeda acknowledges support by JST Moonshot R\&D Grant Number JPMJMS226C.

\nocite{apsrev42Control}
\bibliographystyle{0_MTQ_apsrev4-2_corrected}

\bibliography{refs}

\clearpage

\onecolumngrid
\appendix

\section{The Elements in the Span of the Choi Operator of Unitary Operations}\label{ap:span}

\begin{lem} \label{lem:span1}
Let $\mathcal{I}$ and $\mathcal{O}$ be linear spaces isomorphic to $\mathbb{C}^d$, $U:\mathcal{I}\to\mathcal{O}$ be an arbitrary unitary operator and ${J_U:=\sum_{ij} \ketbra{i}{j}\otimes U\ketbra{i}{j}U^\dagger}$ its associated Choi operator as a unitary CPTP map.
A linear operator $C\in \mathcal{L}(\mathcal{I}\otimes\mathcal{O})$ belongs to the linear space spanned by the Choi operators $J_U$ of unitary CPTP maps if and only if
	\begin{align} \label{Eq:unital}
		\Tr_{\mathcal{I}}[C] = \Tr[C] \, \frac{I_d}{d}~\quad and~ \quad
    \Tr_{\mathcal{O}}[C] = \Tr[C] \, \frac{I_d}{d}.
	\end{align}
\end{lem}
\begin{proof}
	By definition, if $C$ belongs to the linear span of $J_U$, there exist real coefficients $\{c_i\}_i$ and a set of unitary operators $U_i$ such that
\begin{equation}
	C=\sum_i c_i J_{U_i}.
\end{equation}
	Since for any unitary operator $U$ we have that $\Tr_{\mathcal{I}}[J_U]=\Tr_{\mathcal{O}}[J_U]= I_d$, it follows that
\begin{align}
	\Tr[C]&=\sum_i c_i \Tr[J_{U_i}]=d \sum_i c_i,  \\
	\Tr_{\mathcal{I}}[C]&= \left(\sum_i c_i \right) I_d, \\
	\Tr_{\mathcal{O}}[C]&= \left(\sum_i c_i \right) I_d,
\end{align}
hence Eqs.~\eqref{Eq:unital} hold.

	For the converse, we start by pointing out that every self adjoint operator $C\in \mathcal{L}(\mathcal{I}\otimes\mathcal{O})$ can be written as
\begin{equation} \label{eq:HSbasis}
	C=\gamma_0 \,\frac{I_d}{d}\otimes \frac{I_d}{d}
	+ \sum_{i=1}^{d^2-1} \,\alpha_i \sigma_i \otimes \frac{I_d}{d}
	+ \sum_{o=1}^{d^2-1 }\, \beta_o  \frac{I_d}{d} \otimes \sigma_o
	+ \sum_{i=1}^{d^2-1}\sum_{o=1}^{d^2-1} \gamma_{io} \,\sigma_i \otimes \sigma_o,
\end{equation}
where $\{\sigma_i\}_{i=1}^{d^2-1}$ is a set of self adjoint operators on $\mathbb{C}_d$ respecting $\Tr[\sigma_i]=0$ and $\Tr[\sigma_i\,\sigma_j]=\delta_{ij}$. An explicit basis may be obtained via the generalized Gell-Mann matrices~\cite{bertlmann08}. The orthogonality and traceless properties of operators $\sigma_i$ ensure that
\begin{align} \label{Eq:HS1}
	\gamma_0 &= \Tr[C],  \hspace*{20mm} \gamma_{io} = \Tr\Big[C\, (\sigma_i\otimes \sigma_o) \Big], \\ \label{Eq:HS2}
	\alpha_i &= \Tr\Big[ C\, (\sigma_i\otimes I_d) \Big], \quad \quad
	\beta_o = \Tr\Big[ C\, (I_d \otimes \sigma_o) \Big].
\end{align}

The linear span of unital channels is identical to the linear span of unitary CPTP maps \cite{mendl09} (see Ref.~\cite{chiribella20} for a constructive proof).
The Choi operators $C$ of unital CP maps are characterized by $\Tr_{\mathcal{I}}[C]= \Tr[C] \, \frac{I_d}{d}$ and $\Tr_{\mathcal{O}}[C] = \Tr[C] \, \frac{I_d}{d}$, which implies that $\alpha_i=\beta_o=0$ for every $i$ and $o$.
Therefore, an arbitrary linear operator $C$ belongs to the linear span of $J_U$ if it can be written as
\begin{equation}
	C=\gamma_0 \frac{I_d}{d}\otimes \frac{I_d}{d}
	+ \sum_{i=1}^{d^2-1}\sum_{o=1}^{d^2-1} \gamma_{io}\, \sigma_i \otimes \sigma_o.
\end{equation}
\end{proof}

\section{The proof of Lemma~\ref{lem:id_unique}}\label{ap:proof_id}

Consider the identity supermap $\ttsmap{C}$ that transforms an input unitary operation
$\tmap{U}$ on $\ical \otimes \ocal \cong \mathbb{C}^d \otimes \mathbb{C}^d$ to
$\tmap{U}$ on $\pcal \otimes \fcal \cong \mathbb{C}^d \otimes \mathbb{C}^d$.
We show that this supermap transforms an arbitrary input operation $\tmapf{\Lambda}$ on $\ical \otimes \ocal$ to itself on $\pcal \otimes \fcal$
by assuming the completely CP preserving property.
That is, the Choi operator $C$ of the identity supermap satisfies
\begin{align}
C = \sum_{i,j,k,l=1}^d \ketbra{ij}{kl} \otimes \tmap{C}(\ketbra{ij}{kl}) \in \lcal( \ical \otimes \ocal \otimes \pcal \otimes \fcal),
\end{align}
where $\tmap{C}$ denotes the corresponding map of $C$, i.e., $\tmap{C}(J) = \Tr_{\ical \ocal} [C J^t]$ for $J \in \lcal (\ical \otimes \ocal)$,
and we show that this Choi operator is uniquely determined as
\begin{align}
C_0 = \sum_{i,j,k,l=1}^d \ketbra{ij}{kl} \otimes \ketbra{ij}{kl} \in \lcal( \ical \otimes \ocal \otimes \pcal \otimes \fcal)
\end{align}
under the following assumptions:
its action on unitary operations is determined as $\Tr_{\ical \ocal} [ C J_U^t ] = \Tr_{\ical \ocal} [ C_0 J_U^t ]$ for an arbitrary unitary operation $U$;
and it is a positive operator $C \geq 0$.
Note that the action of $\tmap{C}$ is given by $\tmap{C}(J_U) = \Tr_{\ical \ocal} [C J_U^t] = \Tr_{\ical \ocal} [ C_0 J_U^t ] = J_U$ by assumption.
We first consider the diagonal elements of $C$, then consider the off-diagonal elements of $C$.

(First part: diagonal elements of $C$)\quad
Since $\sum_{i,j=1}^d \proj{ij} = I \in \mathrm{span}\{J_U\}$, the action of $\tmap{C}$ is determined as $\tmap{C}(I) = I$, and we obtain
\begin{align}
\sum_{i,j=1}^d \proj{ij} = \tmap{C} (\sum_{i,j=1}^d \proj{ij}) &= \sum_{i,j=1}^d \tmap{C} (\proj{ij}).
\end{align}
Thus, for $\ket{ijkl} \in \ical \otimes \ocal \otimes \pcal \otimes \fcal$
\begin{align}
\sum_{i,j=1}^d \bra{ijkl} C \ket{ijkl} &= \sum_{i,j=1}^d \bra{kl}  \Tr_{\mathcal{IO}}\, (C \proj{ij}^t)\ket{kl} = \sum_{i,j=1}^d \bra{kl} \tmap{C} (\proj{ij}) \ket{kl} \notag \\
&= \bra{kl} \sum_{i,j=1}^d \tmap{C} (\proj{ij}) \ket{kl} = \bra{kl} \big( \sum_{i,j=1}^d \proj{ij}  \big) \ket{kl} = 1, \label{sum_equal_1}
\end{align}
holds for all $k,l$, and we obtain
\begin{align}
\Tr\, [C] = \sum_{i,j,k,l=1}^d \bra{ijkl} C \ket{ijkl} &= d^2 \label{sum_equal_d2}.
\end{align}
Since every diagonal element of a positive semidefinite operator is non-negative,
Eq.~\eqref{sum_equal_1} shows that every diagonal element of $C$, say $a$, satisfies $0 \leq a \leq 1$,
and Eq.~\eqref{sum_equal_d2} shows that the sum of all diagonal elements of $C$ is equal to $d^2$.

Since $\ketbra{ij}{ji},\ketbra{ii}{jj} \in \mathrm{span} \{ J_U \}$ holds for all $ i \neq j$ from Lemma~\ref{lem:span1}, the action of $\tmap{C}$ is determined as
\begin{align}
\tmap{C} ( \ketbra{ij}{ji} ) = \ketbra{ij}{ji}, \quad \tmap{C} ( \ketbra{ii}{jj} ) = \ketbra{ii}{jj},
\end{align}
and the Choi operator $C$ satisfies
\begin{align}
\bra{ijij} C \ket{jiji} = 1, \quad \bra{iiii} C \ket{jjjj} = 1. \label{off_diagonal_of_choi}
\end{align}

Since $C \geq 0$ all principal minors of $C$ are non-negative.
Consider the following two-dimensional submatrices of $C$
\begin{align}
\begin{pmatrix}
\bra{ijij} C \ket{ijij} & \bra{ijij} C \ket{jiji} \\
\bra{jiji} C \ket{ijij} & \bra{jiji} C \ket{jiji}
\end{pmatrix},
\quad
\begin{pmatrix}
\bra{iiii} C \ket{iiii} & \bra{iiii} C \ket{jjjj} \\
\bra{jjjj} C \ket{iiii} & \bra{jjjj} C \ket{jjjj}
\end{pmatrix},
\end{align}
with $i\neq j$.
Eq.~\eqref{off_diagonal_of_choi} indicates that all the off-diagonal elements of these matrices are 1,
and the diagonal elements are between 0 and 1 due to Eq.~\eqref{sum_equal_1}.
Then these matrices are positive semi-definite only if
\begin{align}
\bra{ijij} C \ket{ijij} = 1, \quad \bra{iiii} C \ket{iiii}=1 \label{diagonal_of_choi}
\end{align}
for all $i \neq j$.

Since the sum of all diagonal elements is equal to $d^2$ due to Eq.~\eqref{sum_equal_d2},
and Eq.~\eqref{diagonal_of_choi} provides the $ d(d-1) + d =d^2$ elements that are equal to 1,
the other diagonal elements of $C$ are equal to 0.
For positive semi-definite matrices, if a diagonal element is 0, all elements in that row or column are also 0.
Thus, we have
\begin{align}
\bra{i_1 j_1 k_1 l_1} C \ket{i_2 j_2 k_2 l_2} \neq 0
\end{align}
only if $(i_1,j_1,k_1,l_1) = (i,j,i,j)$ and $(i_2,j_2,k_2,l_2) = (k,l,k,l)$ for some $i,j,k,l$.

(Second part: off-diagonal elements of $W$)\quad
The second part of the proof is to show that all the elements which are non-zero are equal to 1.
Consider a Choi operator $J \in \mathrm{span}\{J_U\}$ with $\bra{ij} J \ket{kl} \neq 0$ for all $i\neq j, k\neq l$,
e.g., the Choi operator of the Fourier transform.
Since $J \in \mathrm{span}\{J_U\}$, the equality
\begin{align}
\bra{ij} \tmap{C}(J) \ket{kl} &= \bra{ij} J \ket{kl} \label{eq:cijkl_1}
\end{align}
holds for all $i\neq j, k\neq l$.
On the other hand, the action determined by the Choi operator $C$ is given by
\begin{align}
\bra{ij} \tmap{C}(J) \ket{kl} &= \bra{ij} \Tr_{\ical \ocal} [C (J^t_{\ical \ocal} \otimes I_{\pcal\fcal})] \ket{kl} \\
&=  \Tr [ \bra{ij} C \ket{kl} J^t ]\\
&= \sum_{i',j',k',l'} \bra{i'j'ij} C \ket{k'l'kl} \bra{k'l'} J^t \ket{i'j'} \\
&= \bra{ijij} C \ket{klkl} \bra{kl} J^t \ket{ij} \\
&= \bra{ijij} C \ket{klkl} \bra{ij} J \ket{kl}, \label{eq:cijkl_2}
\end{align}
where the fourth equality holds due to Eq.~\eqref{diagonal_of_choi}.
Therefore, by considering the condition that both Eq.~\eqref{eq:cijkl_1} and Eq.~\eqref{eq:cijkl_2} hold, we obtain $\bra{ijij} C \ket{klkl} = 1$ for all $i\neq j, k\neq l$.
Similarly, consider a Choi operator $J \in \mathrm{span}\{J_U\}$ with $\bra{ii} J \ket{kl} \neq 0$ for all $i$ and $k\neq l$, e.g., the Choi operator of the Fourier transform,
we obtain $\bra{iiii} C \ket{klkl} = 1$ for all $i$ and $k\neq l$.
Similar calculations show that
$\bra{ijji} C \ket{kkkk} = 1$ for $i \neq j$ and all $k$,
and $\bra{iiii} C \ket{kkkk} = 1$ for all $i$ and $k$.
Thus, we uniquely determine all the elements of $C$ and obtain $C = C_0$.

\section{The Proof of Lemma~\ref{lem:qs_unique}}\label{ap:proof_qs}

We start the section by recalling that
\begin{equation}
\displaystyle{W_0:= \sum_{i,j,k=0}^{d-1} ( \ket{ijjkik00} + \ket{jkijik11} )( \bra{ijjkik00} + \bra{jkijik11} )},
\end{equation}
where the Hilbert spaces are ordered as $\ical_1 \otimes \ocal_1 \otimes \ical_2 \otimes \ocal_2 \otimes \pcal_T  \otimes \fcal_T \otimes \pcal_C \otimes \fcal_C$. Hence, for any $J = \ketbra{ijkl}{i'j'k'l'} \in \lcal( \ical_1 \otimes \ocal_1 \otimes \ical_2 \otimes \ocal_2)$, we have that
\begin{align}
\Tr_{in} [W_0 J^t] =& \bra{ijkl} W_0 \ket{i'j'k'l'}  \\
=& \delta_{jk}\delta_{j'k'} \ketbra{il00}{i'l'00} + \delta_{il}\delta_{i'l'} \ketbra{kj11}{k'j'11}
 + \delta_{jk}\delta_{i'l'} \ketbra{il00}{k'j'11} + \delta_{il}\delta_{j'k'} \ketbra{kj11}{i'l'00}. \label{eq:qs_action_cbasis}
\end{align}

We now remark that the matrix elements of $W_0$ are 0 or 1.
In order to prove Theorem~\ref{thm:qs_unique}, it is enough to prove that each element of $W$ for any quantum-switch-like supermap satisfies $| (W)_{ij} | = 0 \mbox{ or } 1$.
This is because if there exists a quantum-switch-like supermap $W$ (or its Choi operator, to be exact) such that satisfies $| (W)_{ij} | = 0 \mbox{ or } 1$ but $W \neq W_0$,
then $W_\epsilon := \epsilon W + (1-\epsilon) W_0$ is also a candidate for a quantum-switch-like supermap and must have a value satisfying $0 < | (W_{\epsilon})_{ij} | < 1$,
which contradicts with the assertion that $W$ for any quantum switch-like supermap satisfies $| (W)_{ij} | = 0 \mbox{ or } 1$.

The rest of the proof is divided into two parts:
the first part considers the diagonal elements of $W$, and the second part considers the off-diagonal elements of $W$.

(First part: diagonal elements of $W$)\quad
Consider $J = \ketbra{ijkl}{i'j'k'l'}$.
For all $i,j,k,l$, there exist $i',j',k',l'$ such that $J \in \mathrm{span} \{ J_{U_1} \otimes J_{U_2} \}$.
This can be shown by noticing that if $i=j$ we can set $i'=j'\neq i$ to ensure that $\ketbra{ij}{i'j'} \in \mathrm{span} \{ J_U \}$, and in this case, we relabel the variables $i',j'$ as $j$.
Similarly, if $i \neq j$ we set $i'=j$ and $j'=i$ so that $\ketbra{ij}{i'j'} \in \mathrm{span} \{ J_U \}$ (see Sec.~\ref{sec:span} or Lemma~\ref{lem:span1}).
The same arguments hold for $k,l$.  Hence, for $ i \neq j$ and $k \neq l$, 
\begin{align}
&\Tr_{in} [W (\ketbra{ijkl}{jilk})^t]
= ( \delta_{jk} \ket{il00} + \delta_{il} \ket{kj11} ) ( \delta_{il} \bra{jk00} + \delta_{jk} \bra{li11}  ) \label{eq:unique_w_action1} \\
&\Tr_{in} [W (\ketbra{ijkk}{jill})^t] 
= ( \delta_{jk} \ket{ik00} + \delta_{ik} \ket{kj11} ) ( \delta_{il} \bra{jl00} + \delta_{jl} \bra{li11}  ) \label{eq:unique_w_action2} \\
&\Tr_{in} [W (\ketbra{iikl}{jjlk})^t]
= ( \delta_{ik} \ket{il00} + \delta_{il} \ket{ki11} ) ( \delta_{jl} \bra{jk00} + \delta_{jk} \bra{lj11}  ) \label{eq:unique_w_action3} \\
&\Tr_{in} [W (\ketbra{iikk}{jjll})^t]
= ( \delta_{ik} \ket{ik00} + \delta_{ik} \ket{ki11} ) ( \delta_{jl} \bra{jl00} + \delta_{jl} \bra{lj11}  ), \label{eq:unique_w_action4}
\end{align}
Also, consider that $I \otimes I \in \mathrm{span} \{ J_{U_1} \otimes J_{U_2} \}$,
we obtain
\begin{gather}
\Tr [\Tr_{in} [W (I \otimes I)^t]] = \Tr [W] = 2 d^3. \label{eq:unique_w_action5}
\end{gather}

In Appendix~\ref{ap:diag}, we show how Eqs.~\eqref{eq:unique_w_action1}~--~\eqref{eq:unique_w_action5} and the positivity $W \geq 0$ implies that, for  $ i \neq j, k \neq l$, the operator $W$ has to respect
\begin{align}
\bra{ijklil00} W \ket{ijklil00} = 1 &\text{ for } j = k \label{eq:unique_w_diag_result_1} \\
\bra{ijklkj11} W \ket{ijklkj11} = 1 &\text{ for } i = l \\
\bra{ijkkik00} W \ket{ijkkik00} = 1 &\text{ for } j = k \\
\bra{ijkkkj11} W \ket{ijkkkj11} = 1 &\text{ for } i = k \\
\bra{iiklil00} W \ket{iiklil00} = 1 &\text{ for } i = k \\
\bra{iiklki11} W \ket{iiklki11} = 1 &\text{ for } i = l \\
\bra{iikkik00} W \ket{iikkik00} = 1 &\text{ for } i = k \\
\bra{iikkki11} W \ket{iikkki11} = 1 &\text{ for } i = k, \label{eq:unique_w_diag_result_8}
\end{align}
or equivalently,
\begin{align}
\bra{ijjlil00} W \ket{ijjlil00} &= 1 \text{ for } i \neq j, j \neq l \\
\bra{ijkikj11} W \ket{ijkikj11} &= 1 \text{ for } i \neq j, k \neq i \\
\bra{ijjjij00} W \ket{ijjjij00} &= 1 \text{ for } i \neq j \\
\bra{ijiiij11} W \ket{ijiiij11} &= 1 \text{ for } i \neq j \\
\bra{iiilil00} W \ket{iiilil00} &= 1 \text{ for } i \neq l \\
\bra{iikiki11} W \ket{iikiki11} &= 1 \text{ for } i \neq k \\
\bra{iiiiii00} W \ket{iiiiii00} &= 1  \\
\bra{iiiiii11} W \ket{iiiiii11} &= 1 .
\end{align}

The number of terms in 1st, 3rd, 5th, and 7th lines (or 2nd, 4th, 6th, 8th lines) is given by
$d(d-1)^2, d(d-1), d(d-1), d $ respectively, and the total number of terms in these 8 lines is $2 d^3$.
Since $\Tr\, [W] = 2 d^3$, all the other diagonal terms of $W$ are 0.

\begin{rem}
Here we illustrate the key idea of the proof in Appendix~\ref{ap:diag} with a simpler example.
Consider a linear map $\tmap{W} : \lcal(\mathbb{C}^2) \to \lcal(\mathbb{C}^2)$ and the corresponding Choi operator $W \in \lcal ( \mathbb{C}^4 )$ satisfying,
\begin{align}
\Tr_{in} [W (\ketbra{i}{j})^t] &= \ketbra{i}{j} \quad \mbox{for} \quad (i,j)=(0,1) \text{ and }\; (i,j)=(1,0) \\
\Tr\, [W] &= 2 \\
W &\geq 0,
\end{align}
then the matrix $W$ has the form of
\begin{align}
W =
\begin{pmatrix}
a & b & 0 & 1 \\
c & d & 0 & 0 \\
0 & 0 & e & f \\
1 & 0 & g & h \\
\end{pmatrix}
\end{align}
with $a+d+e+h = 2$.
From the positivity $W \geq 0$, the diagonal elements have to respect $a,d,e,h \geq 0$, and since the determinant of all its minors has to be non-negative, we obtain $ah-1 \geq 0$. By combining these constraints we can verify that,
\begin{align}
	2 &= a+d+e+h \\
	  & \geq 2 \sqrt{ah} + d + e, \quad &&\left[ \text{since } a+h\geq 2\sqrt{ah}   \right] \\
	  & \geq 2+ e + d, \quad &&\left[ \text{since } ah\geq 1  \right] \\
	  & \geq 2. \quad &&\left[ \text{since } e\geq0, \; d\geq0  \right]
\end{align}
Hence, we obtain $ a+h = 2\sqrt{ah}$ as a necessary condition.
Thus, we obtain $a=h=1$ and the diagonal terms are uniquely determined as $(a,d,e,h) = (1,0,0,1)$.
\end{rem}

(Second part: off-diagonal elements of $W$)\quad
In this part, we show that every element of $W$ satisfies $| (W)_{ij} | = 0 \mbox{ or } 1$.
From the positivity, if a diagonal element is 0, then all the elements of that row or column are 0.
Therefore, by analysing Eq.~\eqref{eq:qs_action_cbasis} and Lemma \ref{lem:norm_equality} below, we see that the action of $W$ can be described by
\begin{align}
\bra{ijkl} W \ket{i'j'k'l'}
&=\quad a_{ijkli'j'k'l'} \delta_{jk}\delta_{j'k'} \ketbra{il00}{i'l'00} \notag \\
&\quad+ b_{ijkli'j'k'l'} \delta_{il}\delta_{i'l'} \ketbra{kj11}{k'j'11} \notag \\
&\quad + c_{ijkli'j'k'l'} \delta_{jk}\delta_{i'l'} \ketbra{il00}{k'j'11} \notag \\
&\quad+ d_{ijkli'j'k'l'} \delta_{il}\delta_{j'k'} \ketbra{kj11}{i'l'00} \label{eq:qs_action_abcd}
\end{align}
with $|a_{ijkli'j'k'l'}|, |b_{ijkli'j'k'l'}|, |c_{ijkli'j'k'l'}|, |d_{ijkli'j'k'l'}| \leq 1$.
As we mentioned, it is enough to show that all the absolute values of them (with the non-vanishing Kronecker deltas) are 1.

In order to show that the coefficients satisfies $|a_{ijkli'j'k'l'}|, |b_{ijkli'j'k'l'}|, |c_{ijkli'j'k'l'}|, |d_{ijkli'j'k'l'}| = 1$ for all $i,j,k,l,i',j',k',l'$ (for those that yield the non-vanishing Kronecker deltas in Eq.~\eqref{eq:qs_action_abcd}), it will be helpful to define a class of linear operators belonging to the $\mathrm{span} \{ J_{U_1} \otimes J_{U_2} \}$ satisfying useful properties. Before introducing this class, we point out that since the operators $\ketbra{ijkl}{i'j'k'l'}$ form a basis, any operator $J$ acting in $\mathcal{I}_1\otimes\mathcal{O}_1\otimes\mathcal{I}_2\otimes\mathcal{O}_2$ can be written as
\begin{equation}
J = \sum_{ijkli'j'k'l'} s_{ijkli'j'k'l} \ketbra{ijkl}{i'j'k'l'},
\end{equation}
for some scalar coefficients $s_{ijkli'j'k'l}$. It will be useful to focus on the case where we set scalar coefficients to be $s_{ijkli'j'k'l}\in\{-1,0,1\}$ and $J$ is guaranteed to be in the linear span $\mathrm{span} \{ J_{U_1} \otimes J_{U_2} \}$.
\begin{lem}\label{lem:norm_equality}
Let $a,b,c,d,a',b',c',d'$ be integers in $\{0,\ldots,d-1\}$.
For any given operator $\ketbra{abcd}{a'b'c'd'}$, there exists an operator $J_{s} \in \mathrm{span} \{ J_{U_1} \otimes J_{U_2} \} $ such that
\begin{align}
J_{s} := \sum_{ijkli'j'k'l'} s_{ijkli'j'k'l} \ketbra{ijkl}{i'j'k'l'},  \label{eq:def_j_gs}
\end{align}
with $s_{ijkli'j'k'l} \in \mathbb{C}$ and $s_{abcda'b'c'd'} \neq 0$.

Additionally, there exist $J_s$ such that further satisfy
\begin{align}
\| \Tr_{in} [W J_s^t] \|_1 &= \sum_{ijkli'j'k'l'}  \| s_{ijkli'j'k'l} \bra{ijkl} W_0 \ket{i'j'k'l'} \|_1 \label{eq:norm_equality_ijkl}.
\end{align}
\end{lem}
The proof of Lemma~\ref{lem:norm_equality} is presented in App.~\ref{ap:offdiag}.

Using the parametrization for $W$ introduced in Eq.~\eqref{eq:qs_action_abcd} and the orthogonality of the control qubits, any operator $J_{s}$ respects
\small
\begin{align}
\| \Tr_{in} (W J_{s}^t) \|_1
&= \| \sum_{ijkli'j'k'l'} s_{ijkli'j'k'l}  \bra{ijkl} W \ket{i'j'k'l'}  \|_1 \\
&\leq \| \sum_{ijkli'j'k'l'} s_{ijkli'j'k'l} \; a_{ijkli'j'k'l'} \delta_{jk}\delta_{j'k'} \ketbra{il00}{i'l'00} \|_1 \notag \\
&\quad + \| \sum_{ijkli'j'k'l'} s_{ijkli'j'k'l}\; b_{ijkli'j'k'l'} \delta_{il}\delta_{i'l'} \ketbra{kj11}{k'j'11} \|_1 \notag \\
&\quad + \| \sum_{ijkli'j'k'l'} s_{ijkli'j'k'l}\; c_{ijkli'j'k'l'} \delta_{jk}\delta_{i'l'} \ketbra{il00}{k'j'11} \|_1  \notag \\
&\quad + \| \sum_{ijkli'j'k'l'} s_{ijkli'j'k'l} \;d_{ijkli'j'k'l'} \delta_{il}\delta_{j'k'} \ketbra{kj11}{i'l'00} \|_1 \nonumber \\
&\leq  \sum_{ijkli'j'k'l'}  |s_{ijkli'j'k'l} | \| \delta_{jk}\delta_{j'k'} \ketbra{il00}{i'l'00} \|_1 +  \sum_{ijkli'j'k'l'}  |s_{ijkli'j'k'l} | \|  \delta_{il}\delta_{i'l'} \ketbra{kj11}{k'j'11} \|_1 \label{eq:one_norm_sign} \\
&\quad \quad + \sum_{ijkli'j'k'l'}   |s_{ijkli'j'k'l} | \| \delta_{jk}\delta_{i'l'} \ketbra{il00}{k'j'11} \|_1 \notag +  \sum_{ijkli'j'k'l'}  |s_{ijkli'j'k'l} | \|  \delta_{il}\delta_{j'k'} \ketbra{kj11}{i'l'00} \|_1,
\end{align}
\normalsize
where the last inequality follows from the relations $|a_{ijkli'j'k'l'}| \leq 1$, $|b_{ijkli'j'k'l'}| \leq 1$, $|c_{ijkli'j'k'l'}| \leq 1$, and $|d_{ijkli'j'k'l'}| \leq 1$.

Also, if the operator $J_{s}$ respects Eq.~\eqref{eq:norm_equality_ijkl} and belongs to $\mathrm{span} \{ J_{U_1} \otimes J_{U_2} \}$, Eq.~\eqref{eq:qs_action_cbasis} ensures that
\begin{align}
\| \Tr_{in} [W J_{s}^t] \|_1
&= \sum_{ijkli'j'k'l'}  \| s_{ijkli'j'k'l} \bra{ijkl}W_0 \ket{i'j'k'l'} \|_1\\
&= \sum_{ijkli'j'k'l'}  |s_{ijkli'j'k'l} | \Big(\| \delta_{jk}\delta_{j'k'} \ketbra{il00}{i'l'00} \|_1 + \delta_{il}\delta_{i'l'} \ketbra{kj11}{k'j'11} \|_1 \notag \\
&\qquad\qquad\qquad\qquad\qquad+ \delta_{jk}\delta_{i'l'} \ketbra{il00}{k'j'11} \|_1 +\delta_{il}\delta_{j'k'} \ketbra{kj11}{i'l'00} \|_1 \Big) \label{eq:one_norm_nosign} ,
\end{align}
quantity which is identical to the r.h.s. of the inequality in \eqref{eq:one_norm_sign}.

We will now show that the absolute value for any all  coefficients $a_{ijkli'j'k'l'}$ with non-vanishing $\delta_{jk}\delta_{j'k'}$ is exactly one.
For that, we pick a particular coefficient $a_{abcda'b'c'd'}$ with non-vanishing $\delta_{ab}\delta_{a'b'}$, assume that $|a_{abcda'b'c'd'}| < 1$, and exhibit a contradiction.
Let $J_{s}$ be an operator where  $s_{abcda'b'c'd'} = 1$ and satisfies the condition in Lemma~\ref{lem:norm_equality}.
Note that the original statement in Lemma~\ref{lem:norm_equality} is that $s_{abcda'b'c'd'} \neq 0$,
but here the operator $s_{abcda'b'c'd'}^{-1} J_s$ is also in the span $J_{s} \in \mathrm{span} \{ J_{U_1} \otimes J_{U_2} \}$ and satisfies $s_{abcda'b'c'd'} =1$.
Since $|a_{ijkli'j'k'l'}| < 1$ for all indices, it holds that
\begin{align}
\| \sum_{ijkli'j'k'l'} s_{ijkli'j'k'l} \; a_{ijkli'j'k'l'} \delta_{jk}\delta_{j'k'} \ketbra{il00}{i'l'00} \|_1
&\leq \sum_{ijkli'j'k'l'} \| s_{ijkli'j'k'l} \; a_{ijkli'j'k'l'} \delta_{jk}\delta_{j'k'} \ketbra{il00}{i'l'00} \|_1 \notag \\
&< \sum_{ijkli'j'k'l'}  |s_{ijkli'j'k'l} | \| \delta_{jk}\delta_{j'k'} \ketbra{il00}{i'l'00} \|_1,
\end{align}
where the second inequality is strict because
\begin{align}
\| s_{abcda'b'c'd'} \; a_{abcda'b'c'd'} \delta_{ab} \delta_{a'b'} \ketbra{ad00}{a'd'00} \|_1 < |s_{abcda'b'c'd'} | \| \delta_{ab} \delta_{a'b'} \ketbra{ad00}{a'd'00} \|_1
\end{align}
is strict.
Hence, inequality~\eqref{eq:one_norm_sign} is also strict, and by combining this inequality with Eq.~\eqref{eq:one_norm_nosign} we have
\begin{align}
\| \Tr_{in} [W J_{s}^t] \|_1
&< \sum_{ijkli'j'k'l'}  |s_{ijkli'j'k'l} | \Big(\| \delta_{jk}\delta_{j'k'} \ketbra{il00}{i'l'00} \|_1 + \delta_{il}\delta_{i'l'} \ketbra{kj11}{k'j'11} \|_1 \notag \\
&\qquad + \delta_{jk}\delta_{i'l'} \ketbra{il00}{k'j'11} \|_1 +\delta_{il}\delta_{j'k'} \ketbra{kj11}{i'l'00} \|_1 \Big) \nonumber \\
&= \| \Tr_{in}(W J_{s}^t) \|_1,
\end{align}
which is a logical contradiction.

Using the same method, we can also show that, for all indices $\{i,j,k,l,i',j',k',l'\}$, we have $|b_{ijkli'j'k'l'}| = 1$, $|c_{ijkli'j'k'l'}| = 1$, and $|d_{ijkli'j'k'l'}| = 1$ what finishes the proof.

\section{Evaluation of the Diagonal Terms in the Proof of Lemma~\ref{lem:qs_unique}}\label{ap:diag}

Here we provide the proof of Eqs.~\eqref{eq:unique_w_diag_result_1}~--~\eqref{eq:unique_w_diag_result_8}
from Eqs.~\eqref{eq:unique_w_action1}~--~\eqref{eq:unique_w_action5} and the positivity $W \geq 0$.
Let us define the disjoint sets,
\begin{gather}
S_1 = \{ \ket{ijklil00} \mid i \neq j, k \neq l, j = k \} \cup \{ \ket{ijklkj11} \mid i \neq j, k \neq l, i = l \} \\
S_2 = \{ \ket{ijkkik00} \mid i \neq j, j = k \} \\
S_3 = \{ \ket{ijkkkj11} \mid i \neq j, i = k \} \\
S_4 = \{ \ket{iiklil00} \mid k \neq l, i = k \} \\
S_5 = \{ \ket{iiklki11} \mid k \neq l, i = l \} \\
S_6 = \{ \ket{iikkik00} \mid i = k \} \\
S_7 = \{ \ket{iikkki11} \mid i = k \},
\end{gather}
and it is enough to prove that
\begin{align}
\prod_{\ket{\psi} \in S_a} \bra{\psi} W \ket{\psi} \geq 1 \label{eq:unique_s_ket_prod_0}
\end{align}
holds for $a=1,\ldots,7$.
This is because if Eq.~\eqref{eq:unique_s_ket_prod_0} holds, the inequality
\begin{align}
\Tr [W] = 2d^3 \geq \sum_{a=1}^7 \sum_{\ket{\psi} \in S_a} \bra{\psi} W \ket{\psi}
\geq \sum_{a=1}^7 |S_a| \left( \prod_{\ket{\psi} \in S_a} \bra{\psi} W \ket{\psi} \right)^{\frac{1}{|S_a|}}
\geq \sum_{a=1}^7 |S_a| = 2d^3
\end{align}
also holds, and by considering the condition for equality, we obtain
\begin{align}
\bra{\psi} W \ket{\psi} = 1
\end{align}
for all $\ket{\psi} \in S_a$ for all $a$.
Here we only prove Eq.~\eqref{eq:unique_s_ket_prod_0} for $a=1,2,6$.
The case of $a=4$ can be proved in the same way as the case of, $a=2$ by considering the symmetry between the two input operations.
The rest cases $a=3,5,7$ can be proved in the same way as $a=2,4,6$, respectively, by considering the symmetry between the action of the control qubit being 0 and 1.

(Case 1: $S_1$) \quad
In this case, we show the following inequality holds
\begin{align}
\prod_{\ket{\psi} \in S_1} \bra{\psi} W \ket{\psi}
&= \left( \prod_{ i \neq j, k \neq l, j = k} \bra{ijklil00} W \ket{ijklil00} \right) \left( \prod_{ i \neq j, k \neq l, i = l} \bra{ijklkj11} W \ket{ijklkj11} \right) \notag \\
&= \left( \prod_{ i \neq k, k \neq l} \bra{ikklil00} W \ket{ikklil00} \right) \left( \prod_{ i \neq k, k \neq l} \bra{kilkli11} W \ket{kilkli11} \right) \notag  \\
&= \prod_{ i \neq k, k \neq l} \bra{ikklil00} W \ket{ikklil00} \bra{kilkli11} W \ket{kilkli11}
\geq 1 . \label{eq:unique_s_ket_prod_1}
\end{align}
Let $\ket{\psi_{ikl}^{0}} := \ket{ikklil00} $ and $\ket{\psi_{ikl}^1} := \ket{kilkli11}$,
then $\bra{\psi_{ikl}^{0}} W \ket{\psi_{ikl}^1} = 1$ holds for arbitrary $i,k,l$ satisfying $i\neq k, k \neq l$ because of Eq.~\eqref{eq:unique_w_action1}.
By considering the positivity of the submatrix
\begin{align}
\begin{pmatrix}
\bra{\psi_{ikl}^{0}} W \ket{\psi_{ikl}^{0}} & \bra{\psi_{ikl}^{0}} W \ket{\psi_{ikl}^1} \\
\bra{\psi_{ikl}^1} W \ket{\psi_{ikl}^{0}} & \bra{\psi_{ikl}^1} W \ket{\psi_{ikl}^1}
\end{pmatrix}
=
\begin{pmatrix}
\bra{\psi_{ikl}^{0}} W \ket{\psi_{ikl}^{0}} & 1 \\
1 & \bra{\psi_{ikl}^1} W \ket{\psi_{ikl}^1}
\end{pmatrix}
\geq 0,
\end{align}
we obtain the inequality
\begin{align}
\bra{\psi_{ikl}^{0}} W \ket{\psi_{ikl}^{0}} \bra{\psi_{ikl}^1} W \ket{\psi_{ikl}^1} \geq 1,
\end{align}
which completes the proof of Eq.~\eqref{eq:unique_s_ket_prod_1}.

(Case 2: $S_2$) \quad
In this case, we show the following inequality holds
\begin{align}
\prod_{\ket{\psi} \in S_3} \bra{\psi} W \ket{\psi}
&= \prod_{ i \neq j, j = k} \bra{ijkkik00} W \ket{ijkkik00} \notag \\
&= \prod_{ i \neq k} \bra{\psi_{ik}} W \ket{\psi_{ik}}
= \prod_{ i < k} \bra{\psi_{ik}} W \ket{\psi_{ik}} \bra{\psi_{ki}} W \ket{\psi_{ki}}
\geq 1, \label{eq:unique_s_ket_prod_2}
\end{align}
where we define $\ket{\psi_{ik}} :=  \ket{ikkkik00} $ for $i,k=0,\ldots,d-1$.
Since $\bra{\psi_{ik}} W \ket{\psi_{ki}} = 1$ holds arbitrary $i,k$ satisfying $i \neq k$ because of Eq.~\eqref{eq:unique_w_action2},
we consider the positivity of the submatrix
\begin{align}
\begin{pmatrix}
\bra{\psi_{ik}} W \ket{\psi_{ik}} & \bra{\psi_{ik}} W \ket{\psi_{ki}} \\
\bra{\psi_{ki}} W \ket{\psi_{ik}} & \bra{\psi_{ki}} W \ket{\psi_{ki}}
\end{pmatrix}
=
\begin{pmatrix}
\bra{\psi_{ik}} W \ket{\psi_{ik}} & 1 \\
1 & \bra{\psi_{ki}} W \ket{\psi_{ki}}
\end{pmatrix}
\geq 0,
\end{align}
and obtain
\begin{align}
\bra{\psi_{ik}} W \ket{\psi_{ik}} \bra{\psi_{ki}} W \ket{\psi_{ki}}  \geq 1
\end{align}
for any $i \neq k$,
which completes the proof of Eq.~\eqref{eq:unique_s_ket_prod_2}.

(Case 3: $S_6$) \quad
In this case, we show the following inequality holds
\begin{align}
\prod_{\ket{\psi} \in S_7} \bra{\psi} W \ket{\psi}
= \prod_{ i = k} \bra{iikkik00} W \ket{iikkik00}
= \prod_{k} \bra{kkkkkk00} W \ket{kkkkkk00}
\geq 1. \label{eq:unique_s_ket_prod_3}
\end{align}
Let $\ket{\psi_k} := \ket{kkkkkk00}$ for $k=0,\ldots,d-1$,
then $\bra{\psi_k} W \ket{\psi_{k'}} = 1$ holds for $k \neq k'$ because of Eq.~\eqref{eq:unique_w_action4}.
By considering the positivity of the submatrix
\begin{align}
&\begin{pmatrix}
\bra{\psi_0} W \ket{\psi_0} & \bra{\psi_0} W \ket{\psi_1} & \cdots & \bra{\psi_0} W \ket{\psi_{d-1}} \\
\bra{\psi_1} W \ket{\psi_0} & \bra{\psi_1} W \ket{\psi_1} & \ddots & \bra{\psi_1} W \ket{\psi_{d-1}} \\
\vdots & \ddots & \ddots & \vdots \\
\bra{\psi_{d-1}} W \ket{\psi_0} & \bra{\psi_{d-1}} W \ket{\psi_1} & \cdots & \bra{\psi_{d-1}} W \ket{\psi_{d-1}}
\end{pmatrix}
\geq 0,
\end{align}
where all the off-diagonal elements are 1,
we obtain Eq.~\eqref{eq:unique_s_ket_prod_3}.

In total, we uniquely determined the diagonal elements of $W$.

\section{Proof of Lemma~\ref{lem:norm_equality}}\label{ap:offdiag}

In order to prove Lemma~\ref{lem:norm_equality} it will be convenient to introduce the following disjoint subsets of the elements of $\mathrm{span}\{ J_U\}$.
\begin{align*}
G_{1a} &= \{ \ketbra{ij}{i'j'} \mid i \neq j \neq i' \neq j' \} \\
G_{1b} &= \{ \ketbra{ij}{i'j'} \mid i = j, i \neq i' \neq j' \mbox{ or } i=j', i \neq i' \neq j \\
&\qquad \mbox{ or } j=i', j \neq i \neq j' \mbox{ or } i'=j', i' \neq i \neq j \} \\
G_{1c} &= \{ \ketbra{ij}{i'j'} \mid  i=j, i'=j', i \neq i' \mbox{ or } i=j',j=i', i \neq j  \}, \\
G_2 &= \{ \sum_{j=i+k} \ketbra{ij}{ij} \mid k=0, \ldots, d-1 \}, \\
G_{3a} &= \{ \ketbra{ij}{ii} - \ketbra{jj}{ji}  \mid i \neq j \}, \\
G_{3b} &= \{ \ketbra{ji}{ii} - \ketbra{jj}{ij} \mid i \neq j \}, \\
G_{4a} &= \{ \ketbra{ij}{ik} - \ketbra{jj}{jk}  \mid i \neq j \neq k \}, \\
G_{4b} &= \{ \ketbra{ij}{kj} - \ketbra{ii}{ki} \mid i \neq j \neq k \}.
\end{align*}
Here $G_{3b}$ is the swapped version of $G_{3a}$ in the sense that every element of $G_{3b}$ corresponds to an element in $G_{3a}$ with the two systems swapped.
Similarly, $G_{4b}$ is the swapped version of $G_{4a}$.
We also denote $G_1 = G_{1a} \cup G_{1b} \cup G_{1c}$, $G_3 = G_{3a} \cup G_{3b}$, and $G_4 = G_{4a} \cup G_{4b}$.  The elements are of the form $\ketbra{ij}{kl}$ or their linear combination.  We refer to $i, j, k, l$ as \textit{indices}.
Note that
\begin{itemize}
\item the elements with all the four indices taking a distinct value (such as $0123$) are in $G_{1a}$,
\item the elements where the indices are composed of three distinct values (such as $0012$) are in $G_{1b}$ and $G_4$,
\item the elements where the indices are composed of two distinct values and two indices for each value (such as $0011$) are in $G_{1c}$ and $G_2$,
\item the elements where the three indices take the same value while the remaining one takes another value (such as $0001$) are in $G_3$,
\item the elements with the indices taking the same value (such as $0000$) are in $G_2$.
\end{itemize}

One can check that, for any given operator $\ketbra{ab}{a'b'}$, there is an operator $J_\alpha$ from the set $\cup_{a=1,2,3,4} G_a $ in which $\ketbra{ab}{a'b'}$ appears with a non-zero coefficient.
Thus, for any given operator $\ketbra{abcd}{a'b'c'd'}$, there is an operator $J_\alpha \otimes J_\beta$ from the set $\cup_{a,b=1,2,3,4} G_a \otimes G_b $ in which $\ketbra{abcd}{a'b'c'd'}$ appears with a non-zero coefficient.

In order to prove Lemma~\ref{lem:norm_equality}, it is enough to show that for any operator $J_s \in \cup_{a,b=1,2,3,4} G_a \otimes G_b$,
\begin{align}
\| \Tr_{in} [W J_s^t] \|_1 &= \sum_{ijkli'j'k'l'}  \| s_{ijkli'j'k'l} \bra{ijkl} W_0 \ket{i'j'k'l'} \|_1  \quad (\mbox{the same as Eq.}~\eqref{eq:norm_equality_ijkl}) \notag
\end{align}
holds.
In the remaining, we consider all possible cases of $G_a \otimes G_b$ and show that Eq.~\eqref{eq:norm_equality_ijkl} holds for each $J_s \in G_a \otimes G_b$.
Note that it is enough to consider 10 cases that satisfy $a \leq b$ because of the symmetry between the two systems.

(Case 1: $G_1 \otimes G_1$) \quad
In this case, the Eq.~\eqref{eq:norm_equality_ijkl} is trivial because every $J_s \in G_1 \otimes G_1$ has only one element,
that is, for every $J_s \in G_1 \otimes G_1$, $s_{abcda'b'c'd'} \neq 0$ for only one $\{ abcda'b'c'd' \}$,
and thus
\begin{align}
\| \Tr_{in} [W J_{s}^t] \|_1 = \| s_{abcda'b'c'd} \bra{abcd} W_0 \ket{a'b'c'd'} \|_1 = \sum_{ijkli'j'k'l'} \| s_{ijkli'j'k'l'}  \bra{ijkl} W_0 \ket{i'j'k'l'} \|_1
\end{align}

(Case 2: $G_1 \otimes G_2$) \quad
In this case, every possible $J_s$ can be written as $\ketbra{ab}{a'b'} \otimes \sum_{l=k+m} \ketbra{kl}{kl}$ for some $\{ a, b, a', b', m\}$,
and thus
\begin{align}
\| \Tr_{in} [W J_{s}^t] \|_1
&= \| \sum_{k=0}^{d-1} \bra{abk,k+m} W_0 \ket{a'b'k,k+m} \|_1 \\
&= \| \sum_{k=0}^{d-1} (\delta_{jk}\delta_{j'k} \ketbra{i,k+m,00}{i',k+m,00} + \delta_{i,k+m}\delta_{i',k+m} \ketbra{kj11}{kj'11} \notag \\
&\qquad + \delta_{jk}\delta_{i',k+m} \ketbra{i,k+m,00}{kj'11} + \delta_{i,k+m}\delta_{j'k} \ketbra{kj11}{i',k+m,00}) \|_1 \\
&= \sum_{k=0}^{d-1} (\| \delta_{jk}\delta_{j'k} \ketbra{i,k+m,00}{i',k+m,00} + \delta_{i,k+m}\delta_{i',k+m} \ketbra{kj11}{kj'11}\notag \\
&\qquad + \delta_{jk}\delta_{i',k+m} \ketbra{i,k+m,00}{kj'11} + \delta_{i,k+m}\delta_{j'k} \ketbra{kj11}{i',k+m,00} \|_1 ) \\
&= \sum_{k=0}^{d-1} \|  \bra{abk,k+m} W_0 \ket{a'b'k,k+m} \|_1
\end{align}
where the second equality holds because all the elements are non-negative.

(Case 3: $G_2 \otimes G_2$) \quad
Similar to Case 2, expanding $\| \Tr_{in} [W J_{s}^t] \|_1$ provides only non-negative elements, and thus Eq.~\eqref{eq:norm_equality_ijkl} holds.

Before proceeding to the other cases,
we first point out that almost all the calculations are just the verification of Eq.~\eqref{eq:norm_equality_ijkl}.
Moreover, the procedure for these calculations are the same, namely,
\begin{enumerate}
\item using the difference of the states of control qubits to expand the norm,
\item using inequalities like $i \neq j$ to expand the norm in two ways:
\begin{enumerate}
\item the difference of the states: for example, if $i \neq j$, then $\| \proj{i} - \proj{j} \|_1 = \| \proj{i} \|_1 + \| \proj{j} \|_1$
\item the difference of non-zero Kronecker delta: for example, if $i \neq j$, then $\| \delta_{ik} - \delta_{jk} \|_1 = \| \delta_{ik} \|_1 + \| \delta_{jk} \|_1 $
\end{enumerate}
\item all the terms become non-negative at this point, and regrouping the terms, we can obtain the r.h.s. of Eq.~\eqref{eq:norm_equality_ijkl}
\end{enumerate}

(Case 4: $G_1 \otimes G_3$) \quad
Here we consider $ G_3 $ with only the elements of the form $ \ketbra{kl}{kk} - \ketbra{ll}{lk} $, that is, $G_{3a}$.
The proof for $G_{3b}$ follows from the proof for $G_{3a}$ as follows.
Consider $F = F_1 \otimes F_2 \otimes F_T \otimes X_C$, $W_0$ satisfies $W_0 = F W_0 F$, and
\begin{align}
\| \Tr_{in} [W (F_1 J_s F_1)^t] \|_1
&= \| \Tr_{in} [W_0 (F_1 J_s F_1)^t] \|_1 \\
&= \| \Tr_{in} [(F_1 W_0 F_1) J_s^t] \|_1 \\
&= \| \Tr_{in} [(F_2 \otimes F_T \otimes X_C) W_0 (F_2 \otimes F_T \otimes X_C) J_s^t] \|_1 \\
&= \| (F_T \otimes X_C) \Tr_{in} [(F_2 W_0 F_2) J_s^t ] (F_T \otimes X_C)\|_1 \\
&= \| \Tr_{in} [(F_2 W_0 F_2) J_s^t ] \|_1 \\
&= \| \Tr_{in} [W_0 (F_2 J_s F_2)^t] \|_1
\end{align}
holds,
where the fifth equality holds because the flip operation and Pauli X operation only swaps the position of matrix elements and does not change the element-wise 1-norm.
Thus, showing the equation for $J_s \in G_1 \otimes G_{3b}$ is equivalent to showing the equation for $J_s \in (F G_1 F) \otimes G_{3a} = G_1 \otimes G_{3a}$

In this case, every possible $J_s$ can be written as $\ketbra{ij}{i'j'} \otimes (\ketbra{kl}{kk} - \ketbra{ll}{lk})$ for some $\{ i,j,i',j',k,l\}$, and thus
\begin{align}
\| \Tr_{in} [W J_s^t] \|_1 %
&=  \| \bra{ijkl} {W_0} \ket{i'j'kk} - \bra{ijll} {W_0} \ket{i'j'lk}  \|_1 \\
&= \| (\delta_{jk}\delta_{j'k} \ketbra{il00}{i'k00} - \delta_{jl}\delta_{j'l} \ketbra{il00}{i'k00}) \notag \\
&\quad+ (\delta_{il}\delta_{i'k} \ketbra{kj11}{kj'11} - \delta_{il}\delta_{i'k} \ketbra{lj11}{lj'11}) \notag \\
&\quad+ (\delta_{jk}\delta_{i'k} \ketbra{il00}{kj'11} - \delta_{jl}\delta_{i'k} \ketbra{il00}{lj'11} ) \notag \\
&\quad+ (\delta_{il}\delta_{j'k} \ketbra{kj11}{i'k00} - \delta_{il}\delta_{j'l} \ketbra{lj11}{i'k00}) \|_1 \\
&= \| \delta_{jk}\delta_{j'k} \ketbra{il00}{i'k00} - \delta_{jl}\delta_{j'l} \ketbra{il00}{i'k00} \|_1 \notag \\
&\quad+ \| \delta_{il}\delta_{i'k} \ketbra{kj11}{kj'11} - \delta_{il}\delta_{i'k} \ketbra{lj11}{lj'11} \|_1 \notag \\
&\quad+ \| \delta_{jk}\delta_{i'k} \ketbra{il00}{kj'11} - \delta_{jl}\delta_{i'k} \ketbra{il00}{lj'11} \|_1 \notag \\
&\quad+ \| \delta_{il}\delta_{j'k} \ketbra{kj11}{i'k00} - \delta_{il}\delta_{j'l} \ketbra{lj11}{i'k00} \|_1 \\
&= \| \delta_{jk}\delta_{j'k} \ketbra{il00}{i'k00} \|_1 + \| \delta_{jl}\delta_{j'l} \ketbra{il00}{i'k00} \|_1 \notag \\
&\quad+ \| \delta_{il}\delta_{i'k} \ketbra{kj11}{kj'11} \|_1 + \| \delta_{il}\delta_{i'k} \ketbra{lj11}{lj'11} \|_1 \notag \\
&\quad+ \| \delta_{jk}\delta_{i'k} \ketbra{il00}{kj'11} \|_1 + \| \delta_{jl}\delta_{i'k} \ketbra{il00}{lj'11} \|_1 \notag \\
&\quad+ \| \delta_{il}\delta_{j'k} \ketbra{kj11}{i'k00} \|_1 + \| \delta_{il}\delta_{j'l} \ketbra{lj11}{i'k00} \|_1 \\
&= \| \delta_{jk}\delta_{j'k} \ketbra{il00}{i'k00} + \delta_{il}\delta_{i'k} \ketbra{kj11}{kj'11}  \notag \\
&\qquad+  \delta_{jk}\delta_{i'k} \ketbra{il00}{kj'11} + \delta_{il}\delta_{j'k} \ketbra{kj11}{i'k00} \|_1 \notag \\
&\quad+ \| \delta_{jl}\delta_{j'l} \ketbra{il00}{i'k00} + \delta_{il}\delta_{i'k} \ketbra{lj11}{lj'11} \notag \\
&\qquad+ \delta_{jl}\delta_{i'k} \ketbra{il00}{lj'11} + \delta_{il}\delta_{j'l} \ketbra{lj11}{i'k00} \|_1 \\
&= \| \bra{ijkl} {W_0} \ket{i'j'kk} \|_1 + \| \bra{ijll} {W_0} \ket{i'j'lk}  \|_1 \\
\end{align}
where the third equality holds because the states of control qubits are different,
the forth equality holds because $k \neq l$,
and the fifth equality holds because all the terms are non-negative.

(Case 5: $G_2 \otimes G_3$) \quad
As in the case 4, we consider $G_{3a}$ instead.
In this case, every possible $J_s$ can be written as $\sum_{j=i+m} \ketbra{ij}{ij} \otimes (\ketbra{kl}{kk} - \ketbra{ll}{lk})$ for some $\{m,k,l\}$, and thus
\begin{align}
\| \Tr_{in} [W J_s^t] \|_1 &= \| \sum_{i=0}^{d-1} (\bra{i,i+m,kl} {W_0} \ket{i,i+m,kk} - \bra{i,i+m,ll} {W_0} \ket{i,i+m,lk})  \|_1 \\
&= \| \sum_{i=0}^{d-1}  \delta_{i+m,k}\delta_{i+m,k} \ketbra{il00}{ik00} - \delta_{i+m,l}\delta_{i+m,l} \ketbra{il00}{ik00} \|_1 \notag \\
&\quad + \| \sum_{i=0}^{d-1} \delta_{il}\delta_{ik} \ketbra{k,i+m,11}{k,i+m,11} - \delta_{il}\delta_{ik} \ketbra{l,i+m,11}{l,i+m,11} \|_1 \notag \\
&\quad + \| \sum_{i=0}^{d-1} \delta_{i+m,k}\delta_{ik} \ketbra{il00}{k,i+m,11} - \delta_{i+m,l}\delta_{ik} \ketbra{il00}{l,i+m,11} \|_1  \notag \\
&\quad + \| \sum_{i=0}^{d-1} \delta_{il}\delta_{i+m,k} \ketbra{k,i+m,11}{ik00} - \delta_{il}\delta_{i+m,l} \ketbra{l,i+m,11}{ik00} \|_1
\end{align}
\begin{align}
\phantom{\| \Tr_{in} [W J_s^t] \|_1} &= \| \sum_{i=0}^{d-1}  \delta_{i+m,k}\delta_{i+m,k} \ketbra{il00}{ik00} \|_1 + \| \sum_{i=0}^{d-1}\delta_{i+m,l}\delta_{i+m,l} \ketbra{il00}{ik00} \|_1 \notag \\
&\quad + \| \sum_{i=0}^{d-1} \delta_{il}\delta_{ik} \ketbra{k,i+m,11}{k,i+m,11} \|_1 + \| \sum_{i=0}^{d-1}\delta_{il}\delta_{ik} \ketbra{l,i+m,11}{l,i+m,11} \|_1 \notag \\
&\quad + \| \sum_{i=0}^{d-1} \delta_{i+m,k}\delta_{ik} \ketbra{il00}{k,i+m,11} \|_1 + \| \sum_{i=0}^{d-1}\delta_{i+m,l}\delta_{ik} \ketbra{il00}{l,i+m,11} \|_1  \notag \\
&\quad + \| \sum_{i=0}^{d-1} \delta_{il}\delta_{i+m,k} \ketbra{k,i+m,11}{ik00} \|_1 + \| \sum_{i=0}^{d-1}\delta_{il}\delta_{i+m,l} \ketbra{l,i+m,11}{ik00} \|_1 \\
&= \sum_{i=0}^{d-1} ( \|  \bra{i,i+m,kl} {W_0} \ket{i,i+m,kk}\|_1 + \| \bra{i,i+m,ll} {W_0} \ket{i,i+m,lk}  \|_1 )
\end{align}

Here, the second equality holds because the states of control qubits are different,
the third equality holds because $k \neq l$,
and the fourth equality holds because all the terms are non-negative.

(Case 6: $G_3 \otimes G_3$) \quad
In this case, we divide the first $G_3$ into $G_{3a}$ and $G_{3b}$.
Here we only consider $G_{3a} \otimes G_{3a}$ and $G_{3b} \otimes G_{3a}$,
because the proof for $G_{3b} \otimes G_{3b}$ is equivalent to the proof for $G_{3a} \otimes G_{3a}$,
the proof for $G_{3b} \otimes G_{3a}$ is equivalent to the proof for $G_{3a} \otimes G_{3b}$.

For the case of $G_{3a} \otimes G_{3a}$,
$J_s$ can be written as $(\ketbra{ij}{ii} - \ketbra{jj}{ji}) \otimes (\ketbra{kl}{kk} - \ketbra{ll}{lk})$ for some $\{i,j,k,l\}$, and thus
\begin{align}
&\| \Tr_{in} [W J_s^t] \|_1 \notag \\
&= \| \bra{ijkl} {W_0} \ket{iikk} - \bra{ijll} {W_0} \ket{iilk}
 - \bra{jjkl} {W_0} \ket{jikk} + \bra{jjll} {W_0} \ket{jilk} \|_1 \\
&= \| (\delta_{jk}\delta_{ik} - \delta_{jl}\delta_{il} ) \ketbra{il00}{ik00}
- (\delta_{jk}\delta_{ik} - \delta_{jl}\delta_{il} ) \ketbra{jl00}{jk00} \|_1 \notag \\
&\quad + \| (\delta_{il}\delta_{ik} -\delta_{jl}\delta_{jk}) \ketbra{kj11}{ki11}
- (\delta_{il}\delta_{ik} - \delta_{jl}\delta_{jk} ) \ketbra{lj11}{li11} \|_1 \notag \\
&\quad + \| \delta_{jk}\delta_{ik} \ketbra{il00}{ki11} - \delta_{jl}\delta_{ik} \ketbra{il00}{li11}
-\delta_{jk}\delta_{jk} \ketbra{jl00}{ki11} + \delta_{jl}\delta_{jk} \ketbra{jl00}{li11} \|_1 \notag \\
&\quad + \| \delta_{il}\delta_{ik} \ketbra{kj11}{ik00} - \delta_{il}\delta_{il} \ketbra{lj11}{ik00}
-\delta_{jl}\delta_{ik} \ketbra{kj11}{jk00} + \delta_{jl}\delta_{il} \ketbra{lj11}{jk00} \|_1 \\
&= \| (\delta_{jk}\delta_{ik} - \delta_{jl}\delta_{il}) \ketbra{il00}{ik00} \|_1
+ \| (\delta_{jk}\delta_{ik} - \delta_{jl}\delta_{il}) \ketbra{jl00}{jk00} \|_1 \notag \\
&\quad + \| (\delta_{il}\delta_{ik} -\delta_{jl}\delta_{jk}) \ketbra{kj11}{ki11} \|_1
+ \| (\delta_{il}\delta_{ik} - \delta_{jl}\delta_{jk} ) \ketbra{lj11}{li11} \|_1 \notag \\
&\quad + \| \delta_{jk}\delta_{ik} \ketbra{il00}{ki11} \|_1 + \| \delta_{jl}\delta_{ik} \ketbra{il00}{li11} \|_1
+ \| \delta_{jk}\delta_{jk} \ketbra{jl00}{ki11} \|_1 + \| \delta_{jl}\delta_{jk} \ketbra{jl00}{li11} \|_1 \notag \\
&\quad + \| \delta_{il}\delta_{ik} \ketbra{kj11}{ik00} \|_1 + \| \delta_{il}\delta_{il} \ketbra{lj11}{ik00} \|_1
+ \| \delta_{jl}\delta_{ik} \ketbra{kj11}{jk00} \|_1 + \| \delta_{jl}\delta_{il} \ketbra{lj11}{jk00} \|_1 \\
&= \| \delta_{jk}\delta_{ik} \ketbra{il00}{ik00} \|_1 + \| \delta_{jl}\delta_{il} \ketbra{il00}{ik00} \|_1
+ \| \delta_{jk}\delta_{ik} \ketbra{jl00}{jk00} \|_1 + \| \delta_{jl}\delta_{il} \ketbra{jl00}{jk00} \|_1 \notag \\
&\quad + \| \delta_{il}\delta_{ik} \ketbra{kj11}{ki11} \|_1 + \| \delta_{jl}\delta_{jk} \ketbra{kj11}{ki11} \|_1
+ \| \delta_{il}\delta_{ik} \ketbra{lj11}{li11} \|_1 + \| \delta_{jl}\delta_{jk} \ketbra{lj11}{li11} \|_1 \notag \\
&\quad + \| \delta_{jk}\delta_{ik} \ketbra{il00}{ki11} \|_1 + \| \delta_{jl}\delta_{ik} \ketbra{il00}{li11} \|_1
+ \| \delta_{jk}\delta_{jk} \ketbra{jl00}{ki11} \|_1 + \| \delta_{jl}\delta_{jk} \ketbra{jl00}{li11} \|_1 \notag \\
&\quad + \| \delta_{il}\delta_{ik} \ketbra{kj11}{ik00} \|_1 + \| \delta_{il}\delta_{il} \ketbra{lj11}{ik00} \|_1
+ \| \delta_{jl}\delta_{ik} \ketbra{kj11}{jk00} \|_1 + \| \delta_{jl}\delta_{il} \ketbra{lj11}{jk00} \|_1 \\
&= \| \bra{ijkl} {W_0} \ket{iikk} \|_1 + \| \bra{ijll} {W_0} \ket{iilk} \|_1
 + \| \bra{jjkl} {W_0} \ket{jikk} \|_1 + \| \bra{jjll} {W_0} \ket{jilk} \|_1
\end{align}
where the second equality holds because the states of control qubits are different,
the third and the fourth equality holds because $i \neq j, k \neq l$,
and the fifth equality holds because all the terms are non-negative.

Similarly, an element of $G_{3b} \otimes G_{3a}$ can be written as
$(\ketbra{ji}{ii} - \ketbra{jj}{ij}) \otimes (\ketbra{kl}{kk} - \ketbra{ll}{lk})$ for some $\{i,j,k,l\}$, and thus
\begin{align}
& \| \Tr_{in} [W J_s^t] \|_1 \notag \\
&= \| \bra{jikl} {W_0} \ket{iikk} - \bra{jill} {W_0} \ket{iilk} - \bra{jjkl} {W_0} \ket{ijkk} + \bra{jjll} {W_0} \ket{ijlk} \|_1 \\
&= \| (\delta_{ik}\delta_{ik} - \delta_{il}\delta_{il} -\delta_{jk}\delta_{jk}  + \delta_{jl}\delta_{jl}) \ketbra{jl00}{ik00}) \|_1 \notag\\
&\quad + \| (\delta_{jl}\delta_{ik} \ketbra{ki11}{ki11} - \delta_{jl}\delta_{ik} \ketbra{li11}{li11}
-\delta_{jl}\delta_{ik} \ketbra{kj11}{kj11} + \delta_{jl}\delta_{ik} \ketbra{lj11}{lj11}) \|_1 \notag\\
&\quad + \| (\delta_{ik}\delta_{ik} \ketbra{jl00}{ki11} - \delta_{il}\delta_{ik} \ketbra{jl00}{li11}
-\delta_{jk}\delta_{ik} \ketbra{jl00}{kj11} + \delta_{jl}\delta_{ik} \ketbra{jl00}{lj11} ) \|_1 \notag\\
&\quad + \| (\delta_{jl}\delta_{ik} \ketbra{ki11}{ik00} - \delta_{jl}\delta_{il} \ketbra{li11}{ik00}
-\delta_{jl}\delta_{jk} \ketbra{kj11}{ik00} + \delta_{jl}\delta_{jl} \ketbra{lj11}{ik00} \|_1 \\
&= \| (\delta_{ik}\delta_{ik} - \delta_{il}\delta_{il} -\delta_{jk}\delta_{jk}  + \delta_{jl}\delta_{jl}) \ketbra{jl00}{ik00}) \|_1 \notag\\
&\quad + \| \delta_{jl}\delta_{ik} \ketbra{ki11}{ki11} \|_1 + \| \delta_{jl}\delta_{ik} \ketbra{li11}{li11} \|_1
+ \| \delta_{jl}\delta_{ik} \ketbra{kj11}{kj11} \|_1 + \| \delta_{jl}\delta_{ik} \ketbra{lj11}{lj11} \|_1 \notag\\
&\quad + \| \delta_{ik}\delta_{ik} \ketbra{jl00}{ki11} \|_1 + \| \delta_{il}\delta_{ik} \ketbra{jl00}{li11} \|_1
+ \| \delta_{jk}\delta_{ik} \ketbra{jl00}{kj11} \|_1 + \| \delta_{jl}\delta_{ik} \ketbra{jl00}{lj11} \|_1 \notag\\
&\quad + \| \delta_{jl}\delta_{ik} \ketbra{ki11}{ik00} \|_1 + \| \delta_{jl}\delta_{il} \ketbra{li11}{ik00} \|_1
+ \| \delta_{jl}\delta_{jk} \ketbra{kj11}{ik00} \|_1 + \| \delta_{jl}\delta_{jl} \ketbra{lj11}{ik00} \|_1 \\
&= \| (\delta_{ik}\delta_{ik} + \delta_{jl}\delta_{jl}) \ketbra{jl00}{ik00}) \|_1 + \| ( \delta_{il}\delta_{il} +\delta_{jk}\delta_{jk} ) \ketbra{jl00}{ik00}) \|_1   \notag\\
&\quad + \| \delta_{jl}\delta_{ik} \ketbra{ki11}{ki11} \|_1 + \| \delta_{jl}\delta_{ik} \ketbra{li11}{li11} \|_1
+ \| \delta_{jl}\delta_{ik} \ketbra{kj11}{kj11} \|_1 + \| \delta_{jl}\delta_{ik} \ketbra{lj11}{lj11} \|_1 \notag\\
&\quad + \| \delta_{ik}\delta_{ik} \ketbra{jl00}{ki11} \|_1 + \| \delta_{il}\delta_{ik} \ketbra{jl00}{li11} \|_1
+ \| \delta_{jk}\delta_{ik} \ketbra{jl00}{kj11} \|_1 + \| \delta_{jl}\delta_{ik} \ketbra{jl00}{lj11} \|_1 \notag\\
&\quad + \| \delta_{jl}\delta_{ik} \ketbra{ki11}{ik00} \|_1 + \| \delta_{jl}\delta_{il} \ketbra{li11}{ik00} \|_1
+ \| \delta_{jl}\delta_{jk} \ketbra{kj11}{ik00} \|_1 + \| \delta_{jl}\delta_{jl} \ketbra{lj11}{ik00} \|_1 \\
&= \| \bra{jikl} {W_0} \ket{iikk} \|_1 + \| \bra{jill} {W_0} \ket{iilk} \|_1 + \| \bra{jjkl} {W_0} \ket{ijkk} \|_1 + \| \bra{jjll} {W_0} \ket{ijlk} \|_1 \\
\end{align}
where the second equality holds because the states of control qubits are different.
the third and the fourth equality holds because $i \neq j, k \neq l$,
and the fifth equality holds because all the terms are non-negative.

(Case 7: $G_1 \otimes G_4$) \quad
Here we consider $ G_4 $ with only the elements of the form $ \ketbra{ij}{ik} - \ketbra{jj}{jk} $.
The elements of the other form can be evaluated in the same way, since the two input systems are symmetric.
We denote this set as $G_{4a}$, that is,

In this case, every possible $J_s$ can be written as $\ketbra{ij}{i'j'} \otimes (\ketbra{kl}{km} - \ketbra{ll}{lm} )$ for some $\{i,j,i',j',k,l,m\}$, and thus
\begin{align}
\| \Tr_{in} [W J_s^t] \|_1
&= \| \bra{ijkl} {W_0} \ket{i'j'km} - \bra{ijll} {W_0} \ket{i'j'lm}  \|_1 \\
&= \| (\delta_{jk}\delta_{j'k} \ketbra{il00}{i'm00} - \delta_{jl}\delta_{j'l} \ketbra{il00}{i'm00}) \|_1 \notag \\
&\quad+ \| (\delta_{il}\delta_{i'm} \ketbra{kj11}{kj'11} - \delta_{il}\delta_{i'm} \ketbra{lj11}{lj'11}) \|_1 \notag \\
&\quad+ \| (\delta_{jk}\delta_{i'm} \ketbra{il00}{kj'11} - \delta_{jl}\delta_{i'm} \ketbra{il00}{lj'11} ) \|_1 \notag \\
&\quad+ \| (\delta_{il}\delta_{j'k} \ketbra{kj11}{i'm00} - \delta_{il}\delta_{j'l} \ketbra{lj11}{i'm00}) \|_1 \\
&= \| \delta_{jk}\delta_{j'k} \ketbra{il00}{i'm00} \|_1 + \| \delta_{jl}\delta_{j'l} \ketbra{il00}{i'm00} \|_1 \notag \\
&\quad+ \| \delta_{il}\delta_{i'm} \ketbra{kj11}{kj'11} \|_1 + \| \delta_{il}\delta_{i'm} \ketbra{lj11}{lj'11} \|_1 \notag \\
&\quad+ \| \delta_{jk}\delta_{i'm} \ketbra{il00}{kj'11} \|_1 + \| \delta_{jl}\delta_{i'm} \ketbra{il00}{lj'11} \|_1 \notag \\
&\quad+ \| \delta_{il}\delta_{j'k} \ketbra{kj11}{i'm00} \|_1 + \| \delta_{il}\delta_{j'l} \ketbra{lj11}{i'm00} \|_1 \\
&= \| \bra{ijkl} {W_0} \ket{i'j'km} \|_1 + \| \bra{ijll} {W_0} \ket{i'j'lm}  \|_1 \\
\end{align}
where the second equality holds because the states of control qubits are different,
the third equality holds because $k \neq l \neq m$,
and the fourth equality holds because all the terms are non-negative.

(Case 8: $G_2 \otimes G_4$) \quad
As in the case 7, we only consider $G_{4a}$ instead of $G_4$.
In this case, every possible $J_s$ can be written as $\sum_{j=i+n} \ketbra{ij}{ij} \otimes (\ketbra{kl}{km} - \ketbra{ll}{lm} )$ for some $\{n,k,l,m\}$, and thus
\begin{align}
& \| \Tr_{in} [W J_s^t] \|_1 \notag \\
&= \| \sum_{i=0}^{d-1} (\bra{i,i+n,kl} {W_0} \ket{i,i+n,km} - \bra{i,i+n,ll} {W_0} \ket{i,i+n,lm})  \|_1 \\
&= \| \sum_{i=0}^{d-1}  \delta_{i+n,k}\delta_{i+n,k} \ketbra{il00}{im00} - \delta_{i+n,l}\delta_{i+n,l} \ketbra{il00}{im00} \|_1 \notag \\
&\quad + \| \sum_{i=0}^{d-1} \delta_{il}\delta_{im} \ketbra{k,i+n,11}{k,i+n,11} - \delta_{il}\delta_{im} \ketbra{l,i+n,11}{l,i+n,11} \|_1 \notag \\
&\quad + \| \sum_{i=0}^{d-1} \delta_{i+n,k}\delta_{im} \ketbra{il00}{k,i+n,11} - \delta_{i+n,l}\delta_{im} \ketbra{il00}{l,i+n,11} \|_1  \notag \\
&\quad + \| \sum_{i=0}^{d-1} \delta_{il}\delta_{i+n,k} \ketbra{k,i+n,11}{im00} - \delta_{il}\delta_{i+n,l} \ketbra{l,i+n,11}{im00} \|_1 \\
&= \| \sum_{i=0}^{d-1} \delta_{i+n,k}\delta_{i+n,k} \ketbra{il00}{im00} \|_1 + \| \sum_{i=0}^{d-1} \delta_{i+n,l}\delta_{i+n,l} \ketbra{il00}{im00} \|_1 \notag \\
&\quad + \| \sum_{i=0}^{d-1} \delta_{il}\delta_{im} \ketbra{k,i+n,11}{k,i+n,11} \|_1 + \| \sum_{i=0}^{d-1} \delta_{il}\delta_{im} \ketbra{l,i+n,11}{l,i+n,11} \|_1 \notag \\
&\quad + \| \sum_{i=0}^{d-1} \delta_{i+n,k}\delta_{im} \ketbra{il00}{k,i+n,11} \|_1 + \| \sum_{i=0}^{d-1} \delta_{i+n,l}\delta_{im} \ketbra{il00}{l,i+n,11} \|_1  \notag \\
&\quad + \| \sum_{i=0}^{d-1} \delta_{il}\delta_{i+n,k} \ketbra{k,i+n,11}{im00} \|_1 + \| \sum_{i=0}^{d-1}  \delta_{il}\delta_{i+n,l} \ketbra{l,i+n,11}{im00} \|_1 \\
&= \sum_{i=0}^{d-1} (\|  \bra{i,i+n,kl} {W_0} \ket{i,i+n,km} \|_1 + \| \bra{i,i+n,ll} {W_0} \ket{i,i+n,lm}  \|_1 )
\end{align}
where the second equality holds because the states of control qubits are different,
the third equality holds because $k \neq l \neq m$,
and the fourth equality holds because all the terms are non-negative.

(Case 9: $G_3 \otimes G_4$) \quad
In this case, we divide the $G_3$ into $G_{3a}$ and $G_{3b}$, and $G_4$ into $G_{4a}$ and $G_{4b}$.
Here we only consider $G_{3a} \otimes G_{4a}$ and $G_{3b} \otimes G_{4a}$
because the proof for $G_{3b} \otimes G_{4b}$ is equivalent to the proof for $G_{3a} \otimes G_{4a}$,
and the proof for $G_{3a} \otimes G_{4b}$ is equivalent to the proof for $G_{3b} \otimes G_{4a}$.

We first consider the case of $G_{3a} \otimes G_{4a}$,
where $J_s$ can be written as $(\ketbra{ij}{ii} - \ketbra{jj}{ji}) \otimes (\ketbra{kl}{km} - \ketbra{ll}{lm} )$ for some $\{i,j,k,l,m\}$.
In this case,
\begin{align}
&\| \Tr_{in} [W J_s^t] \|_1 \notag \\
&= \| \bra{ijkl} {W_0} \ket{iikm} - \bra{ijll} {W_0} \ket{iilm} - \bra{jjkl} {W_0} \ket{jikm} + \bra{jjll} {W_0} \ket{jilm} \|_1 \\
&=  \| (\delta_{jk}\delta_{ik} - \delta_{jl}\delta_{il}) \ketbra{il00}{im00}
- (\delta_{jk}\delta_{ik} - \delta_{jl}\delta_{il}) \ketbra{jl00}{jm00}) \|_1 \notag \\
&\quad + \| (\delta_{il}\delta_{im} -\delta_{jl}\delta_{jm}) \ketbra{kj11}{ki11}
- (\delta_{il}\delta_{im} - \delta_{jl}\delta_{jm}) \ketbra{lj11}{li11}) \|_1 \notag \\
&\quad + \| (\delta_{jk}\delta_{im} \ketbra{il00}{ki11} - \delta_{jl}\delta_{im} \ketbra{il00}{li11}
-\delta_{jk}\delta_{jm} \ketbra{jl00}{ki11} + \delta_{jl}\delta_{jm} \ketbra{jl00}{li11} ) \|_1 \notag \\
&\quad + \| (\delta_{il}\delta_{ik} \ketbra{kj11}{im00} - \delta_{il}\delta_{il} \ketbra{lj11}{im00}
-\delta_{jl}\delta_{ik} \ketbra{kj11}{jm00} + \delta_{jl}\delta_{il} \ketbra{lj11}{jm00} \|_1 \\
&=  \| (\delta_{jk}\delta_{ik} - \delta_{jl}\delta_{il}) \ketbra{il00}{im00} \|_1
+ \| (\delta_{jk}\delta_{ik} - \delta_{jl}\delta_{il}) \ketbra{jl00}{jm00}) \|_1 \notag \\
&\quad + \| (\delta_{il}\delta_{im} -\delta_{jl}\delta_{jm}) \ketbra{kj11}{ki11} \|_1
+ \| (\delta_{il}\delta_{im} - \delta_{jl}\delta_{jm}) \ketbra{lj11}{li11}) \|_1 \notag \\
&\quad + \| \delta_{jk}\delta_{im} \ketbra{il00}{ki11} \|_1 + \| \delta_{jl}\delta_{im} \ketbra{il00}{li11} \|_1
+ \| \delta_{jk}\delta_{jm} \ketbra{jl00}{ki11} \|_1 + \| \delta_{jl}\delta_{jm} \ketbra{jl00}{li11} \|_1 \notag \\
&\quad + \| \delta_{il}\delta_{ik} \ketbra{kj11}{im00} \|_1 + \| \delta_{il}\delta_{il} \ketbra{lj11}{im00} \|_1
+ \| \delta_{jl}\delta_{ik} \ketbra{kj11}{jm00} \|_1 + \| \delta_{jl}\delta_{il} \ketbra{lj11}{jm00} \|_1 \\
&=  \| \delta_{jk}\delta_{ik} \ketbra{il00}{im00} \|_1 + \| \delta_{jl}\delta_{il} \ketbra{il00}{im00} \|_1
+ \| \delta_{jk}\delta_{ik} \ketbra{jl00}{jm00} \|_1 + \| \delta_{jl}\delta_{il} \ketbra{jl00}{jm00} \|_1 \notag \\
&\quad + \| \delta_{il}\delta_{im} \ketbra{kj11}{ki11} \|_1 + \| \delta_{jl}\delta_{jm} \ketbra{kj11}{ki11} \|_1
+ \| \delta_{il}\delta_{im} \ketbra{lj11}{li11} \|_1 + \| \delta_{jl}\delta_{jm} \ketbra{lj11}{li11} \|_1 \notag \\
&\quad + \| \delta_{jk}\delta_{im} \ketbra{il00}{ki11} \|_1 + \| \delta_{jl}\delta_{im} \ketbra{il00}{li11} \|_1
+ \| \delta_{jk}\delta_{jm} \ketbra{jl00}{ki11} \|_1 + \| \delta_{jl}\delta_{jm} \ketbra{jl00}{li11}  \|_1 \notag \\
&\quad + \| \delta_{il}\delta_{ik} \ketbra{kj11}{im00} \|_1 + \| \delta_{il}\delta_{il} \ketbra{lj11}{im00} \|_1
+ \| \delta_{jl}\delta_{ik} \ketbra{kj11}{jm00} \|_1 + \| \delta_{jl}\delta_{il} \ketbra{lj11}{jm00} \|_1 \\
&= \| \bra{ijkl} {W_0} \ket{iikm} \|_1 + \| \bra{ijll} {W_0} \ket{iilm} \|_1
+ \| \bra{jjkl} {W_0} \ket{jikm} \|_1 + \| \bra{jjll} {W_0} \ket{jilm} \|_1
\end{align}
where the second equality holds because the states of control qubits are different,
the third and the fourth equality holds because $i \neq j$ and $k \neq l \neq m$,
and the fifth equality holds because all the terms are non-negative.

Similarly, $J_s \in G_{3b} \otimes G_{4a}$ can be written as
$(\ketbra{ji}{ii} - \ketbra{jj}{ij}) \otimes (\ketbra{kl}{km} - \ketbra{ll}{lm} )$ for some $\{i,j,k,l,m\}$.
In this case,
\begin{align}
&\| \Tr_{in} [W J_s^t] \|_1 \notag \\
&=\| \bra{jikl} {W_0} \ket{iikm} - \bra{jill} {W_0} \ket{iilm} - \bra{jjkl} {W_0} \ket{ijkm} + \bra{jjll} {W_0} \ket{ijlm} \|_1 \\
&= \| (\delta_{ik}\delta_{ik} - \delta_{il}\delta_{il} -\delta_{jk}\delta_{jk}  + \delta_{jl}\delta_{jl}) \ketbra{jl00}{im00}) \|_1 \notag\\
&\quad + \| (\delta_{jl}\delta_{im} \ketbra{ki11}{ki11} - \delta_{jl}\delta_{im} \ketbra{li11}{li11}
-\delta_{jl}\delta_{im} \ketbra{kj11}{kj11} + \delta_{jl}\delta_{im} \ketbra{lj11}{lj11}) \|_1 \notag\\
&\quad + \| (\delta_{ik}\delta_{im} \ketbra{jl00}{ki11} - \delta_{il}\delta_{im} \ketbra{jl00}{li11}
-\delta_{jk}\delta_{im} \ketbra{jl00}{kj11} + \delta_{jl}\delta_{im} \ketbra{jl00}{lj11} ) \|_1 \notag\\
&\quad + \| (\delta_{jl}\delta_{ik} \ketbra{ki11}{im00} - \delta_{jl}\delta_{il} \ketbra{li11}{im00}
-\delta_{jl}\delta_{jk} \ketbra{kj11}{im00} + \delta_{jl}\delta_{jl} \ketbra{lj11}{im00} \|_1 \\
&= \| (\delta_{ik}\delta_{ik} + \delta_{jl}\delta_{jl}) \ketbra{jl00}{im00} \|_1 + \| (\delta_{il}\delta_{il} +\delta_{jk}\delta_{jk}  ) \ketbra{jl00}{im00} \|_1 \notag\\
&\quad + \| (\delta_{jl}\delta_{im} \ketbra{ki11}{ki11} \|_1 + \| \delta_{jl}\delta_{im} \ketbra{li11}{li11} \|_1
+ \| \delta_{jl}\delta_{im} \ketbra{kj11}{kj11} \|_1 + \| \delta_{jl}\delta_{im} \ketbra{lj11}{lj11}) \|_1 \notag\\
&\quad + \| (\delta_{ik}\delta_{im} \ketbra{jl00}{ki11} \|_1 + \| \delta_{il}\delta_{im} \ketbra{jl00}{li11} \|_1
+ \| \delta_{jk}\delta_{im} \ketbra{jl00}{kj11} \|_1 + \| \delta_{jl}\delta_{im} \ketbra{jl00}{lj11} ) \|_1 \notag\\
&\quad + \| (\delta_{jl}\delta_{ik} \ketbra{ki11}{im00} \|_1 + \| \delta_{jl}\delta_{il} \ketbra{li11}{im00} \|_1
+ \| \delta_{jl}\delta_{jk} \ketbra{kj11}{im00} \|_1 + \| \delta_{jl}\delta_{jl} \ketbra{lj11}{im00} \|_1 \\
&=\| \bra{jikl} {W_0} \ket{iikm} \|_1 + \| \bra{jill} {W_0} \ket{iilm} \|_1 + \| \bra{jjkl} {W_0} \ket{ijkm} \|_1 + \| \bra{jjll} {W_0} \ket{ijlm} \|_1
\end{align}
where the second equality holds because the states of control qubits are different.
the third holds because $i \neq j$ and $k \neq l \neq m$,
and the fourth equality holds because all the terms are non-negative.

(Case 10: $G_4 \otimes G_4$) \quad
Here we only consider $G_{4a} \otimes G_{4a}$ and $G_{4b} \otimes G_{4a}$ as the previous case.

We first consider the case of $G_{4a} \otimes G_{4a}$, which element can be written as
$(\ketbra{ij}{in} - \ketbra{jj}{jn}) \otimes (\ketbra{kl}{km} - \ketbra{ll}{lm} )$ for some $\{i,j,n,k,l,m\}$.
\begin{align}
& \| \Tr_{in} [W J_s^t] \|_1 \notag \\
&= \| \bra{ijkl} {W_0} \ket{inkm} - \bra{ijll} {W_0} \ket{inlm} - \bra{jjkl} {W_0} \ket{jnkm} + \bra{jjll} {W_0} \ket{jnlm} \|_1 \\
&=  \| (\delta_{jk}\delta_{nk} - \delta_{jl}\delta_{nl}) \ketbra{il00}{im00}
- (\delta_{jk}\delta_{nk} - \delta_{jl}\delta_{nl}) \ketbra{jl00}{jm00}) \|_1 \notag \\
&\quad + \| (\delta_{il}\delta_{im} -\delta_{jl}\delta_{jm}) \ketbra{kj11}{kn11}
- (\delta_{il}\delta_{im} - \delta_{jl}\delta_{jm}) \ketbra{lj11}{ln11}) \|_1 \notag \\
&\quad + \| \delta_{jk}\delta_{im} \ketbra{il00}{kn11} - \delta_{jl}\delta_{im} \ketbra{il00}{ln11}
-\delta_{jk}\delta_{jm} \ketbra{jl00}{kn11} + \delta_{jl}\delta_{jm} \ketbra{jl00}{ln11} \|_1 \notag \\
&\quad + \| \delta_{il}\delta_{nk} \ketbra{kj11}{im00} - \delta_{il}\delta_{nl} \ketbra{lj11}{im00}
-\delta_{jl}\delta_{nk} \ketbra{kj11}{jm00} + \delta_{jl}\delta_{nl} \ketbra{lj11}{jm00} \|_1 \\
&=  \| (\delta_{jk}\delta_{nk} - \delta_{jl}\delta_{nl}) \ketbra{il00}{im00} \|_1
+ \| (\delta_{jk}\delta_{nk} - \delta_{jl}\delta_{nl}) \ketbra{jl00}{jm00}) \|_1 \notag \\
&\quad + \| (\delta_{il}\delta_{im} -\delta_{jl}\delta_{jm}) \ketbra{kj11}{kn11} \|_1
+ \| (\delta_{il}\delta_{im} - \delta_{jl}\delta_{jm}) \ketbra{lj11}{ln11}) \|_1 \notag \\
&\quad + \| \delta_{jk}\delta_{im} \ketbra{il00}{kn11} \|_1 + \| \delta_{jl}\delta_{im} \ketbra{il00}{ln11} \|_1
+ \| \delta_{jk}\delta_{jm} \ketbra{jl00}{kn11} \|_1 + \| \delta_{jl}\delta_{jm} \ketbra{jl00}{ln11} \|_1 \notag \\
&\quad + \| \delta_{il}\delta_{nk} \ketbra{kj11}{im00} \|_1 + \| \delta_{il}\delta_{nl} \ketbra{lj11}{im00} \|_1
+ \| \delta_{jl}\delta_{nk} \ketbra{kj11}{jm00} \|_1 + \| \delta_{jl}\delta_{nl} \ketbra{lj11}{jm00} \|_1 \\
&=  \| \delta_{jk}\delta_{nk} \ketbra{il00}{im00} \|_1 + \| \delta_{jl}\delta_{nl} \ketbra{il00}{im00} \|_1
+ \| \delta_{jk}\delta_{nk} \ketbra{jl00}{jm00} \|_1 + \| \delta_{jl}\delta_{nl}\ketbra{jl00}{jm00} \|_1 \notag \\
&\quad + \| \delta_{il}\delta_{im} \ketbra{kj11}{kn11} \|_1 + \| \delta_{jl}\delta_{jm} \ketbra{kj11}{kn11} \|_1
+ \| \delta_{il}\delta_{im} \ketbra{lj11}{ln11} \|_1 + \| \delta_{jl}\delta_{jm} \ketbra{lj11}{ln11} \|_1 \notag \\
&\quad + \| \delta_{jk}\delta_{im} \ketbra{il00}{kn11} \|_1 + \| \delta_{jl}\delta_{im} \ketbra{il00}{ln11} \|_1
+ \| \delta_{jk}\delta_{jm} \ketbra{jl00}{kn11} \|_1 + \| \delta_{jl}\delta_{jm} \ketbra{jl00}{ln11} \|_1 \notag \\
&\quad + \| \delta_{il}\delta_{nk} \ketbra{kj11}{im00} \|_1 + \| \delta_{il}\delta_{nl} \ketbra{lj11}{im00} \|_1
+ \| \delta_{jl}\delta_{nk} \ketbra{kj11}{jm00} \|_1 + \| \delta_{jl}\delta_{nl} \ketbra{lj11}{jm00} \|_1 \\
&= \| \bra{ijkl} {W_0} \ket{inkm} \|_1 + \| \bra{ijll} {W_0} \ket{inlm} \|_1 + \| \bra{jjkl} {W_0} \ket{jnkm} \|_1 + \| \bra{jjll} {W_0} \ket{jnlm} \|_1
\end{align}
where the second equality holds because the states of control qubits are different,
the third and the fourth equality holds because $i \neq j$ and $k \neq l \neq m$,
and the fifth equality holds because all the terms are non-negative.

Similarly, $J_s \in G_{4b} \otimes G_{4a}$ can be written as
$(\ketbra{ji}{ni} - \ketbra{jj}{nj}) \otimes (\ketbra{kl}{km} - \ketbra{ll}{lm} )$ for some $\{i,j,k,l,m\}$.
In this case,
\begin{align}
&\| \Tr_{in} [W J_s^t] \|_1 \notag \\
&= \| \bra{jikl} {W_0} \ket{nikm} - \bra{jill} {W_0} \ket{nilm} - \bra{jjkl} {W_0} \ket{njkm} + \bra{jjll} {W_0} \ket{njlm} \|_1 \\
&= \| (\delta_{ik}\delta_{ik} - \delta_{il}\delta_{il} -\delta_{jk}\delta_{jk}  + \delta_{jl}\delta_{jl}) \ketbra{jl00}{nm00}) \|_1 \notag\\
&\quad + \| \delta_{jl}\delta_{nm} \ketbra{ki11}{ki11} - \delta_{jl}\delta_{nm} \ketbra{li11}{li11}
-\delta_{jl}\delta_{nm} \ketbra{kj11}{kj11} + \delta_{jl}\delta_{nm} \ketbra{lj11}{lj11} \|_1 \notag\\
&\quad + \| \delta_{ik}\delta_{nm} \ketbra{jl00}{ki11} - \delta_{il}\delta_{nm} \ketbra{jl00}{li11}
-\delta_{jk}\delta_{nm} \ketbra{jl00}{kj11} + \delta_{jl}\delta_{nm} \ketbra{jl00}{lj11} \|_1 \notag\\
&\quad + \| \delta_{jl}\delta_{ik} \ketbra{ki11}{nm00} - \delta_{jl}\delta_{il} \ketbra{li11}{nm00}
-\delta_{jl}\delta_{jk} \ketbra{kj11}{nm00} + \delta_{jl}\delta_{jl} \ketbra{lj11}{nm00} \|_1 \\
&= \| (\delta_{ik}\delta_{ik} + \delta_{jl}\delta_{jl}) \ketbra{jl00}{nm00} \|_1 + \| (\delta_{il}\delta_{il} + \delta_{jk}\delta_{jk}) \ketbra{jl00}{nm00} \|_1\notag\\
&\quad + \| \delta_{jl}\delta_{nm} \ketbra{ki11}{ki11} \|_1 + \| \delta_{jl}\delta_{nm} \ketbra{li11}{li11} \|_1
+ \| \delta_{jl}\delta_{nm} \ketbra{kj11}{kj11} \|_1 + \| \delta_{jl}\delta_{nm} \ketbra{lj11}{lj11} \|_1 \notag\\
&\quad + \| \delta_{ik}\delta_{nm} \ketbra{jl00}{ki11} \|_1 + \| \delta_{il}\delta_{nm} \ketbra{jl00}{li11} \|_1
+ \| \delta_{jk}\delta_{nm} \ketbra{jl00}{kj11} \|_1 + \| \delta_{jl}\delta_{nm} \ketbra{jl00}{lj11} \|_1 \notag\\
&\quad + \| \delta_{jl}\delta_{ik} \ketbra{ki11}{nm00} \|_1 + \| \delta_{jl}\delta_{il} \ketbra{li11}{nm00} \|_1
+ \| \delta_{jl}\delta_{jk} \ketbra{kj11}{nm00} \|_1 + \| \delta_{jl}\delta_{jl} \ketbra{lj11}{nm00} \|_1 \\
&= \| \bra{jikl} {W_0} \ket{nikm} \|_1 + \| \bra{jill} {W_0} \ket{nilm} \|_1 + \| \bra{jjkl} {W_0} \ket{njkm} \|_1 + \| \bra{jjll} {W_0} \ket{njlm} \|_1
\end{align}
where the second equality holds because the states of control qubits are different.
The third holds because $i \neq j$ and $k \neq l \neq m$,
and the fourth equality holds because all the terms are non-negative.

\end{document}